\documentclass[12pt]{article}

\title{
  Price Setting on a Network\thanks{
    I would like to thank 
    Federico Boffa,
    Doruk Cetemen,
    Jeff Ely,
    Francesco Giovannoni,
    Gautam Gowrisankaran,
    Sanjeev Goyal,
    Edoardo Grillo,
    Andrei Hagiu,
    Marit Hinnosaar,
    Alexander Matros,
    Ignacio Monz\'{o}n,
    Salvatore Piccolo,
    Fernando Vega-Redondo,
    David Rivers,
    Georg Schaur,
    Cristi\'{a}n Troncoso-Valverde,
    and Julian Wright
    as well as seminar participants at 
    the Collegio Carlo Alberto,
    University of South Carolina, 
    University of Tennessee, 
    Appalachian State University, 
    JEI (Barcelona),
    Lancaster Game Theory Conference,
    University of Bath, 
    Cardiff Business School, 
    Bristol University, and
    First UniBg Economic Theory Symposium
    for their comments and suggestions.
  }
}
\usepackage[utf8]{inputenc}
\usepackage[T1]{fontenc}
\date{\today\thanks{First draft: May 13, 2018. The latest version is at 
    \href{http://toomas.hinnosaar.net/pricing_network.pdf}{\url{toomas.hinnosaar.net/pricing_network.pdf}}
}}

\author{
Toomas Hinnosaar%
\thanks{Collegio Carlo Alberto, \href{mailto:toomas@hinnosaar.net}{\url{toomas@hinnosaar.net}}. %
}%
}
\usepackage{anysize}

\usepackage[T1]{fontenc}
\usepackage{lmodern}
\usepackage{hyperref}
\usepackage{pgf,tikz}
\usepackage{natbib}
\usepackage{enumerate}
\usepackage{amssymb,amsmath,amsthm}
\usepackage{thmtools,thm-restate}
\usepackage{multirow}
\usepackage{graphicx}
\usepackage{subfig}

\newcommand{\dd}[3][]{\frac{d^{#1} #2}{d {#3}^{#1}}}

\newcommand{\R}{\mathbb R}

\newcommand{\N}{\mathbb N}

\usepackage{nicefrac}

\usepackage{cleveref}

\usepackage{color}

\usepackage{thmtools}
\makeatletter
\newcommand{\myref}[1]{\cref{#1}\mynameref{#1}{\csname r@#1\endcsname}}
\newcommand{\Myref}[1]{\Cref{#1}\mynameref{#1}{\csname r@#1\endcsname}}
\def\mynameref#1#2{%
    \begingroup
    \edef\@mytxt{#2}%
    \edef\@mytst{\expandafter\@thirdoffive\@mytxt}%
    \ifx\@mytst\empty\else
    \space(\nameref{#1})\fi
    \endgroup
}
\makeatother

\crefname{assumption}{assumption}{assumptions}
\Crefname{assumption}{Assumption}{Assumptions}
\crefname{condition}{condition}{conditions}
\Crefname{condition}{Condition}{Conditions}
\crefname{figure}{figure}{figures}
\crefname{equation}{equation}{equations}

\newcommand{\bp}{\boldsymbol{p}}

\newcommand{\gN}{\mathcal{N}}

\newcommand{\oP}{\overline{P}}

\newcommand{\madj}{\boldsymbol{A}}
\newcommand{\madjother}{\boldsymbol{B}}
\newcommand{\vone}{\boldsymbol{1}}
\newcommand{\vzero}{\boldsymbol{0}}
\newcommand{\ve}{\boldsymbol{e}}
\newcommand{\infl}{I}
\newcommand{\bonacich}{B}

\newcommand{\Nobs}{\mathcal{O}}
\newcommand{\Nuno}{\mathcal{U}}
\newcommand{\Ninf}{\mathcal{I}}
\newcommand{\pobs}{\boldsymbol{p}}

\usepackage{tikz}
\usetikzlibrary{matrix,chains,positioning,decorations.pathreplacing,arrows,shapes}

\usetikzlibrary{arrows.meta}
\tikzset{>={Latex[width=2mm,length=3mm]},every on chain/.append style={join},every join/.style={->}} 
\usepackage{tabularx}
\usepackage{kbordermatrix}

\declaretheorem[name=Proposition]{proposition}
\declaretheorem[name=Theorem]{theorem}
\declaretheorem[name=Corollary]{corollary}
\declaretheorem[name=Lemma]{lemma}

\declaretheorem[name=Assumption]{assumption}

\begin{document}
\maketitle
\begin{abstract}
Most products are produced and sold by supply chain networks, where an interconnected network of producers and intermediaries set prices to maximize their profits. I show that there exists a unique equilibrium in a price-setting game on a network. The key distortion reducing both total profits and social welfare is multiple-marginalization, which is magnified by strategic interactions. Individual profits are proportional to influentiality, which is a new measure of network centrality defined by the equilibrium characterization. The results emphasize the importance of the network structure when considering policy questions such as mergers or trade tariffs.
\end{abstract}



\section{Introduction}

Most products are produced and sold by supply chain networks, where an interconnected network of producers, suppliers, and intermediaries choose prices to maximize their own profits. Some firms in the network compete with other firms offering similar products or services, some may be regulated, whereas others enjoy market power in their niche and can affect the pricing of the final good.

For example, in the book publishing industry, a publisher purchases content from authors, services from editors and marketing firms, and outsources printing to a printer, who in turn purchases paper, ink, and other supplies from outside companies. The publisher is typically also not selling books directly to consumers but has contracts with distributors, who in turn deal with wholesalers, retail chains, and individual retailers. Many of these companies have some market power and earn positive rents. According to estimates in the New York Times, from a typical \$26 book, the retail side gets about half, printing and transport costs are about 13\%, and the author is generally paid about 15\%, which suggests that the markups to various parties play a major role.\footnote{New York Times ``Math of Publishing Meets the E-Book'', by Motoko Rich (Feb. 28, 2010), \url{https://www.nytimes.com/2010/03/01/business/media/01ebooks.html}.}

The main question of this paper is how are prices set on supply chain networks? In the model, I allow the firms to have market power over their own products and influence the pricing of other monopolists on a network of influences. The characterization needs to overcome two difficulties. First, the demand function is typically non-linear, which means that the optimality conditions are non-linear and therefore recursively solving for best-response functions becomes analytically intractable with most networks. Second, the decisions on the network are interconnected---when firms in different parts of the network make decisions they may have different information and may influence different firms. This means that the problem is dynamic but cannot be solved sequentially, i.e.\ with backward induction. I overcome these problems by characterizing best-responses by their inverses and aggregating all necessary conditions into one necessary condition for equilibrium. I then show that this condition has a unique solution and the implied behavior is indeed an equilibrium.

The main result of the paper is a characterization theorem. Under a few regularity conditions, there exists a unique equilibrium. I provide full characterization for the equilibrium and show how to compute it. The equilibrium condition has a natural interpretation. It equalizes the difference between the equilibrium price of the final good and the total marginal cost with a weighted sum of \emph{influences} in all levels. The most basic level of influence is how marginal price increase influences firms' own profit directly and there is one such influence for each firm. The second level of influence specifies how price increase changes the behavior of firms directly influences by this change. But there are also higher-level influences, as each price change if an influenced firm may have further indirect influences on the behavior of other firms. All those levels of influences are weighed with endogenous weights that are determined by the shape of the demand function and the equilibrium behavior.

A natural question is how social welfare and total profits depend on the network structure? I show that the key distortion that reduces both profits and welfare is multiple-marginalization. Marginalization problem is increasing in the number of firms and is magnified by strategic interactions. The network structure plays a crucial role in determining the efficiency losses. The main takeaway from this is that policy analysis needs to take into account how a particular regulation affects the underlying influences on the supply chain network. For example, a merger can be seemingly efficiency-improving, but if it leads to more influence by the merged firms, it may revert the conclusions about efficiency. Similarly, trade restrictions and tariffs are often designed to change the supply chains. To evaluate the impact of such policies, it is important to take into account their impact on the network of influences.

Another natural question is who are the firms that are more powerful in the network and how does this depend on the network structure? Many firms in the supply chain network have some market power and therefore earn positive profits. But intuitively, firms that affect other firms are more powerful. Indeed, the equilibrium characterization defines a natural measure of influentiality of each firm, which depends on how many other firms the firm influences, but also how influential are the other firms that the firm influences. The firms' markups and profits are proportional to this measure. The measure is a generalization of standard centrality measures, where the weights for the direct and indirect influences are endogenously defined by the demand function. In some special cases, this measure of influentiality reduces to standard measures of network centrality: for example when the demand function is a linear or power function, the influentiality measure simplifies to Bonacich centrality, but I provide other examples which make the measure equivalent to degree centrality or even independent of network structure.

The rest of the paper is structured as follows. 
The next section discusses the related literature. 
\Cref{S:model} introduces the model, illustrates it with some examples of networks, and discusses the regularity assumptions. 
\Cref{S:characterization} provides the main characterization result and describes the main components of the proof---in particular, how the characterization overcomes two main complications that arise from non-linearities in demand and interconnected decisions on the network. 
\Cref{S:marginalization} interprets the characterization by comparing it with some known benchmarks and then discusses the key distortion---multiple-marginalization. 
\Cref{S:influentiality} studies which of the firms are more influential and discusses the relationship between the implied influentiality measure with standard network centrality measures.
\Cref{S:computing} describes how to apply the characterization result to compute the equilibrium and provides further results for some of the most common demand functions (including linear, power, and logit).
\Cref{S:discussion} concludes and discusses policy implications for mergers and trade.
All proofs are in \cref{A:proofs}.

\section{Related Literature}

\paragraph{Network games.}
Price setting on networks and supply chains is extensively studied in several branches of economics literature. Perhaps most directly the paper contributes to the literature on network games, where players take actions on a fixed network and the payoffs depend both on their own and their neighbors' actions. 
According to a survey by \cite{jackson_games_2015}, most works in this literature can be divided into two groups. First, a lot of progress has been made in games with quadratic payoffs (or more generally, payoffs that imply linear best-responses). A seminal paper is \cite{ballester_whos_2006} who show that in such games the equilibrium actions are proportional to Bonacich centrality. \cite{bramoulle_public_2007,calvo-armengol_peer_2009,bramoulle_strategic_2014}; and \cite{allouch_cost_2017} study more general variation of this game and find that Bonacich centrality still determines the equilibrium behavior. \cite{candogan_optimal_2012,bloch_pricing_2013,fainmesser_pricing_2016,ushchev_price_2018} study pricing of goods with network externalities with quadratic payoffs and find that optimal pricing leads to discounts that are proportional to Bonacich centrality. 
\cite{bimpikis_cournot_2019} study Cournot competition on a bipartite network, where the sellers Cournot-compete in markets which they have access to. They show that when the demands are linear and costs quadratic, the equilibrium behavior is proportional to Bonacich centrality.

The second branch of network games studies games with non-quadratic payoffs and is generally able to analyze only qualitative properties of the equilibria rather than provide a full characterization.\footnote{An exception is \cite{choi_trading_2017}, who study price competition on networks, where consumers choose the cheapest paths from source to destination and intermediaries set prices, thus making the game a generalization of Bertrand competition. \cite{choi_trading_2017} provide a full equilibrium characterization of prices without relying on quadratic preferences.}\footnote{
In settings where the players interact on a network randomly, the characterization is often more tractable. This line of research includes, for example, \cite{manea_bargaining_2011,abreu_bargaining_2012} who study bargaining on networks.} A seminal paper is \cite{galeotti_network_2010}.
Compared to these works, in this paper, I provide a characterization result for a game on a network with relatively general payoff structure. The characterization defines a natural new measure of influentiality and the firms' choices and payoffs are proportional to this measure. In special cases when the best-response functions are linear (such as linear demand), this measure is proportional to Bonacich centrality. But as the weights are endogenously defined by the demand function and equilibrium behavior, for all other demand functions the measure of influentiality differs from Bonacich centrality. Indeed, I provide examples of special cases where it can be equivalent to degree centrality or even independent of the network structure.

\paragraph{Industrial organization.}
The paper contributes also to theoretical literature on vertical integration. \cite{spengler_vertical_1950} was the first to describe the double-marginalization problem and after this the literature has extensively studied the benefits and costs of vertical control, including \cite{mathewson_economic_1984,grossman_costs_1986,rey_logic_1986,salinger_vertical_1988,salinger_meaning_1989,riordan_anticompetitive_1998,ordover_equilibrium_1990,farrell_horizontal_1990,bolton_incomplete_1993,kuhn_excess_1999,nocke_vertical_2007}; and \cite{buehler_making_2013}.\footnote{See \cite{vickers_vertical_1991} for a literature review.}
Empirical work shows that production has a network structure \citep{atalay_network_2011}, merger-like behaviors such as airline codeshare agreements may fail to eliminate double-marginalization \citep{gayle_efficiency_2013}, and removing vertical restraints may hurt consumers \citep{luco_vertical_2018}. This literature analyzes many forms of competition and contract structures, but very little is known about networks with more than two levels (upstream-downstream). In this paper, I focus on a simple contract structure (posted prices) and allow only relatively simple competition rules (either price-takers or monopolists), but extend the analysis to general network structures.

\paragraph{Trade and macroeconomics.}
Supply chains are also studied in international trade and macroeconomics.\footnote{For a surveys of strategic trade literature, see \cite{rauch_business_2001} and \cite{spencer_strategic_2008}.}
Empirical work by \cite{bernard_wholesalers_2010,ahn_role_2011,jones_intermediate_2011} shows that intermediaries play a major role in international trade and work by \cite{alfaro_prices_2016} and \cite{conconi_final_2018} show that tariffs and trade barriers have a major impact on the supply chains that are formed.
\cite{chaney_network_2014} documents that international trade has a network structure---firms export only to markets where they have contacts. 
\cite{spencer_vertical_1991} were the first to embed vertical restraints into the trade model,  \cite{antras_intermediated_2011} introduced a trade model where intermediaries reduce search frictions, and \cite{rauch_network_2004} showed that intermediaries may have insufficient incentives for efficient outcomes.
The literature on the network formation in production and trade, such as \cite{oberfield_theory_2018} and \cite{liu_industrial_2018}, is perhaps the closest to the current paper and provides complementary results. It studies network formation, where entrepreneurs choose which inputs to use in their (typically Cobb-Douglas) production technology. In contrast to most of the trade and macroeconomics literature, I study optimal price-setting on a fixed network, where many agents, including the middle-men, have some market power.
  
\paragraph{Supply chain management.}
There is also a specialized supply chain management literature, started by \cite{forrester_industrial_1961}. This literature mostly focuses on other dimensions of supply chain management, but perhaps the closest to the magnified multiple-marginalization problem in the current paper is the idea of the bullwhip effect, introduced by \cite{lee_information_2004}. If firms in the supply chain make sequential choices learning about demand only from orders, then the distortion tends to increase with each additional level of interactions. 
Empirical research by \cite{metters_quantifying_1997} shows the quantitative importance of this effect. \cite{bhattacharya_review_2011} provide a literature review, highlighting 19 different causes documented by the literature. These causes are mostly statistical, behavioral, and practical, and one contribution of the current paper is to add multiple-marginalization to this list---even if all firms in the supply chain network are perfectly informed and fully rational, the monopoly distortions accumulate with more inter-dependencies on networks.
There are a few works that bring related ideas to the supply chain management literature. 
\cite{liu_pricing_2007} study a model similar to this paper, but with just two firms and find that Stackelberg leadership leads to worse outcomes. \cite{perakis_price_2007} study efficiency loss from decentralization (in terms of the price of anarchy) in supply chains. In this paper, I show that the marginalization problem is a general problem for supply chains and is magnified by strategic interactions.

\paragraph{Sequential and aggregative games.}
Methodologically, the paper builds on recent advances in sequential and aggregative games. In a special case, where the firms are making independent decisions (and thus the network is irrelevant), the model is an aggregative game. Aggregative games were first proposed by \cite{selten_preispolitik_1970} and there has been recent progress in aggregative games literature by \cite{jensen_aggregative_2010,martimort_representing_2012}; and \cite{acemoglu_aggregate_2013} that has been used to shed new light on questions in industrial organization by \cite{nocke_multiproduct-firm_2018}.\footnote{For a literature review on aggregative games, see \cite{jensen_aggregative_2017}.
Note that when the number of players is infinitely large, then these games become mean-field games (see for example \cite{jovanovic_anonymous_1988}).} One classical aggregative game is a contest and this paper builds on recent work on sequential contests by \cite{hinnosaar_optimal_2018} extending the methodology to networks and asymmetric costs.\footnote{There are other papers belonging to the intersection of contests and networks literature, including \cite{franke_conflict_2015} and \cite{matros_contests_2018} who study contests on networks and \cite{goyal_competitive_2019} who study contagion on networks.}

\section{Model} \label{S:model}

\subsection{Setup}

The model is static and studies the supply of a single final good. The final good has a demand function $D(P)$, where $P$ is its price. The production and supply process requires $m+n$ inputs. I normalize the units of inputs so that one unit of each input is required to produce one unit of output. 

Input $i$ is produced by firm $i$, that has a constant marginal cost $c_i$ and a price $p_i$ for its product.\footnote{The assumptions that marginal costs are constant and that firms choose prices rather than more complex contracts are restrictive and can be relaxed, but would lead to much less tractable analysis.} The price $p_i$ is firm $i$'s per-unit revenue net of payments to other firms in the model. Due to normalization, the quantity of firm $i$'s product (i.e.\ quantity of input $i$) is equal to $D(P)$. Therefore firm $i$ gets profit $\pi_i(\bp) = (p_i-c_i) D(P)$, where $\bp=(p_1,\dots,p_{m+n})$ and the price of the final good is the sum of all net prices, $P = \sum_{i=1}^{m+n} p_i$.

I assume that $m$ of the inputs $n+1,\dots,n+m$ are produced by \emph{price-takers}, who treat their prices as fixed. Such firm $i$ may operate in a competitive sector or compete as a Bertrand competitor, in which case its price is equal to the marginal cost of the second cheapest firm in this sector. The firm could also operate in a regulated industry and its price is set by a regulator. 

The remaining $n$ firms $1,\dots,n$ are \emph{monopolists}, who set their prices strategically, i.e.\ maximizing profits, anticipating the impact on sales of the final good. To complete the description of the model, I need to specify how the price $p_i$ of firm $i$ affects the behavior of other monopolists, which I do by introducing the network of influences.
Formally, network of influences consists of all $n$ monopolists as nodes and edges that define influences. The edges are described as an $n \times n$ adjacency matrix $\madj$, where an element $a_{ij} = 1$ indicates that firm $i$ \emph{influences} firm $j$. 
That is, when firm $j$ chooses price $p_j$, then it takes price $p_i$ as given and responds optimally to it. Of course, firm $i$ knows this and when choosing $p_i$, it knows that $j$ will respond optimally. Finally, if $i$ and $j$ are not directly linked, i.e.\ $a_{ij}=a_{ji}=0$, then neither responds to deviations by the other firm. They expect the other firm to behave according to its equilibrium strategy. For convenience, I assume that the diagonal elements $a_{ii}=0$. I will discuss a few examples of the network of influences in the next subsection.

Let me make three remarks about the model here. First, the price-takers are non-strategic players, so without loss of generality I replace them by a single parameter $c_0 = \sum_{i=n+1}^{n+m} p_i$. Parameter $c_0$ can be interpreted as a cost for the supply chain. I denote the total per-unit cost to the supply chain by $C = c_0 + \sum_{i=1}^n c_i$.

Second, the analysis does not require that each firm on the supply chain is either always a price-taker or always a monopolist. The assumption that I use in the characterization is that monopolists behave according to their local optimality condition, whereas price-takers take their prices locally as given. It could be that a firm who is a monopolist in one situation becomes a price-taker when the model parameters change.

Third, the network of influences makes the game sequential. If $a_{ij}=1$ then firm $i$ sets its price $p_i$ before firm $j$. Firm $j$ then observes $p_i$ and may respond optimally. Of course, firm $i$ knows this and therefore can anticipate the response of firm $j$. I am looking for a pure-strategy\footnote{Although I am not excluding the possibility of mixed-strategy equilibria, I show that there always exists unique pure-strategy equilibrium, so it is natural to focus on it.} subgame-perfect Nash equilibrium, where players take some of the choices of other players as given and maximize their profits, anticipating the impact on other players' choices and the final good demand.

\subsection{Examples of Network of Influences}

The network of influences I introduce in this paper is related but not the same as the supply-chain network. Typical supply-chain network specifies the flows of goods and services (material flows), as well as the flows of money and information. The specifics of these flows are neither necessary nor sufficient to characterize pricing decisions. For pricing decisions, the model needs to specify what is known to each monopolist at the moment it makes a pricing decision and how it expects this decision to influence the choices of other firms. In other words, the model needs to specify the observability of prices and commitment power of firms. As described above, I model this by assuming that there is a commonly known network, such that whenever there is an edge from $i$ to $j$, firm $j$ observes $p_i$ and therefore takes it into account in its optimization problem. 

Consider first a very simple case with just two firms, $F$ (final goods producer) and $R$ (retailer). Then there are three possible networks, illustrated by \cref{F:twoplayer}. First, case (\cref{F:twoplayer_nolinks}) has no influences, so that firms set their prices $p_F$ and $p_R$ independently and the final good is sold at $P=p_F+p_R$. This could be a reasonable assumption in many cases. For example, if both are big firms that interact with many similar firms. Then the final goods producer $F$ does not best-respond to a particular retailer $R$ but to equilibrium $p_R^*$ of a representative retailer. Similarly, the retailer does not best-respond to deviations by particular producer $F$, but to equilibrium price $p_F^*$ of a representative producer. Another example where it is natural to make this assumption is when two firms are separately selling perfectly complementary products to final consumers. 

Similarly, there could be many reasons for strategic influences. For example, a downstream influence from producer $F$ to retailer $R$ (\cref{F:twoplayer_downstream}) may arise with a large producer and small retailer, where the representative retailer reacts optimally to pricing by $F$. The large producer knows that retailers respond to its pricing and therefore takes into account how a representative retailer best-responds.
Of course, the influence could go in the opposite direction (as in \cref{F:twoplayer_upstream}) for the same reason---a large retailer $R$ knows that a small producer $F$ will best-respond to its price changes. In this paper, I take these influences as given and simply assume that some firms have more commitment power than others for exogenous reasons.
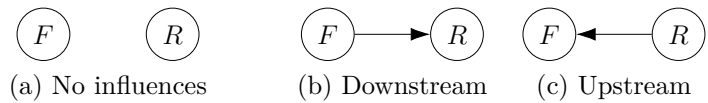
\begin{figure}[!ht]
  \centering
  \subfloat[No influences]{
  \scalebox{0.85}{
      \begin{tikzpicture}
      [
      cnode/.style={draw=black,fill=#1,minimum width=3mm,circle},
      ]
      \node[cnode=white] (nF)   at (0,0)  {$F$};
      \node[cnode=white] (nR)   at (2,0)  {$R$};
      \end{tikzpicture}
  }
  \label{F:twoplayer_nolinks}}
  \qquad
  \subfloat[Downstream]{
  \scalebox{0.85}{
      \begin{tikzpicture}
      [
      cnode/.style={draw=black,fill=#1,minimum width=3mm,circle},
      ]
      \node[cnode=white] (nF)   at (0,0)  {$F$};
      \node[cnode=white] (nR)   at (2,0)  {$R$};
      \draw[->] (nF) -- (nR);
      \end{tikzpicture}
  }
  \label{F:twoplayer_downstream}}
  \subfloat[Upstream]{
  \scalebox{0.85}{
      \begin{tikzpicture}
      [
      cnode/.style={draw=black,fill=#1,minimum width=3mm,circle},
      ]
      \node[cnode=white] (nF)   at (0,0)  {$F$};
      \node[cnode=white] (nR)   at (2,0)  {$R$};
      \draw[<-] (nF) -- (nR);
      \end{tikzpicture}
  }
  \label{F:twoplayer_upstream}}
  \caption{Example: Three possible two-player networks.} \label{F:twoplayer}  
\end{figure}

Let me illustrate the network of influences with three more examples. \Cref{F:strongretailer} depicts an example of a retail chain with downstream-to-upstream influences. In this example, there is a strong retailer $R$, who can commit to adding a markup $p_R$ on top of the wholesale price $P_W$, so that the price of the final good will be $P = P_W +p_r$. The wholesaler $W$ takes $p_R$ as given and commits to its markup $p_W$, so that when the distributor's price is $P_D$, then wholesale price is $P_W = P_D+p_W$ and therefore final good price $P=P_D+p_W+p_R$. Then distributor $D$ sets its markup $p_D$ taking markups $p_W$ and $p_R$ as given. Finally, the final good producer $F$ sets a price $p_F$, taking into account that final consumer will pay $P=p_F+p_D+p_W+p_R$.
\begin {figure}[!ht]
\begin{minipage}{0.5\textwidth}
  \centering
  \scalebox{0.85}{
  \begin{tikzpicture}
  [
  cnode/.style={draw=black,fill=#1,minimum width=3mm,circle},
  ]
%

  \node[cnode=white] (nF)   at (4,-2)   {$F$};

  \node[cnode=white] (nD)   at (6,-2)   {$D$};
  \node[cnode=white] (nW)   at (8,-2)   {$W$};
  \node[cnode=white] (nR)   at (10,-2)  {$R$};

%

  \draw[->] (nR) -- (nW);
  \draw[->] (nW) -- (nD);
  \draw[->] (nD) -- (nF);

  \draw[->,out=130,in=50] (nW) edge (nF);
  \draw[->,out=215,in=325]  (nR) edge (nD);
  \draw[->,out=220,in=320] (nR) edge (nF);
\end{tikzpicture}  
  }
\end{minipage}
\begin{minipage}{0.5\textwidth}
  \[
  \madj = 
  \kbordermatrix{
      & F & D & W & R \\
    F & 0 & 0 & 0 & 0 \\
    D & 1 & 0 & 0 & 0 \\
    W & 1 & 1 & 0 & 0 \\
    R & 1 & 1 & 1 & 0 \\
  }
  \]
\end{minipage}
  \caption{Example: Retail chain with downstream-to-upstream influences.} \label{F:strongretailer}  
\end{figure}
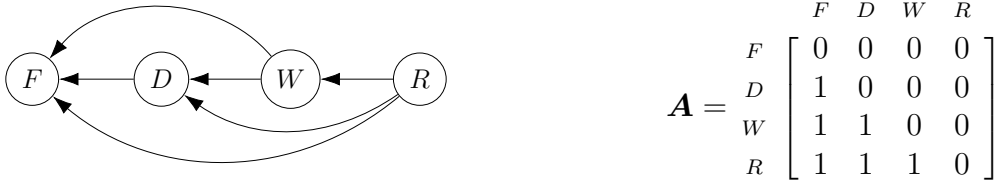

Influences can also go in the opposite direction. \Cref{F:productionchain} gives an example of a production chain. In this example, there is a small producer $F$ who produces final good and uses three inputs, produced by intermediate good producers $I_1, I_2,$ and $I_3$. Firm $F$ takes the prices of its inputs $P_{I_1}, P_{I_2}$, and $P_{I_3}$ as given and chooses the price for the final good, $P_F = P$. 
Intermediate good producer $I_2$ uses two raw materials as inputs, produced respectively by $R_1$ and $R_2$. In this example, firm $I_2$ when choosing $P_{I_2}$ takes $P_{R_1}$ and $P_{R_2}$ as given. Importantly, as firms $I_1$ and $I_3$ do not used these inputs, they do not know the realized prices offered by $R_1$ and $R_2$ but they can make equilibrium conjectures. 
Therefore, while firms $I_1$ and $I_3$ take into account the equilibrium prices $P_{R_i}^*$ in their optimization problem, they cannot respond to potential deviations in $P_{R_i}$, whereas firms $I_2$ and $F$ can and do respond to these deviations.
It is convenient to redefine prices as net prices, net of transfers to the other firms, i.e.\ 
$p_{R_1}=P_{R_1}, p_{R_2}=P_{R_2}, p_{I_1}=P_{I_1}, p_{I_2}= P_{I_2} - p_{R_1}-p_{R_2}$, $p_{I_3}=P_{I_3}$, and $p_F = P - p_{I_1}-p_{I_2}-p_{I_3}-p_{R_1}-p_{R_2}$, so that $\sum_i p_i = P$.
\begin {figure}[!ht]
\begin{minipage}{0.5\textwidth}
  \centering
  \scalebox{0.85}{
  \begin{tikzpicture}
  [
  cnode/.style={draw=black,fill=#1,minimum width=3mm,circle},
  ]
  \node[cnode=white] (nR1) at (0,-0.66) {$R_1$};
  \node[cnode=white] (nR2) at (0,-3.33) {$R_2$};

  \node[cnode=white] (nM1)  at (2,-0.66){$I_1$};
  \node[cnode=white] (nM2)  at (2,-2)   {$I_2$};
  \node[cnode=white] (nM3)  at (2,-3.33){$I_3$};

  \node[cnode=white] (nF)   at (4,-2)   {$F$};


  \draw[->] (nR1) -- (nM2);
  \draw[->] (nR2) -- (nM2);

  \draw[->] (nM1) -- (nF);
  \draw[->] (nM2) -- (nF);
  \draw[->] (nM3) -- (nF);

  \draw[->] (nR1) edge (nF);
  \draw[->] (nR2) edge (nF);


  \end{tikzpicture}  
  }
\end{minipage}
\begin{minipage}{0.5\textwidth}
  \[
  \madj = 
  \kbordermatrix{
        & R_1 & R_2 & I_1 & I_2 & I_3 & F \\
    R_1 & 0 & 0 & 0 & 1 & 0 & 1 \\
    R_2 & 0 & 0 & 0 & 1 & 0 & 1 \\
    I_1 & 0 & 0 & 0 & 0 & 0 & 1 \\
    I_2 & 0 & 0 & 0 & 0 & 0 & 1 \\
    I_3 & 0 & 0 & 0 & 0 & 0 & 1 \\
    F   & 0 & 0 & 0 & 0 & 0 & 0 \\
  }
  \]
\end{minipage}
\caption{Example: Production chain with upstream-to-downstream influences.} \label{F:productionchain}  
\end{figure}
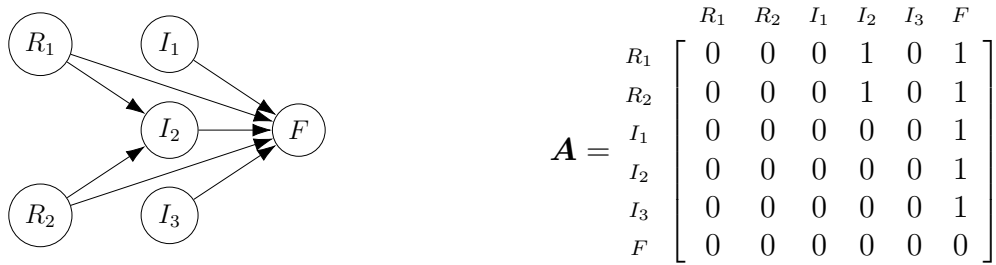

There is no reason to assume that the flows of influence are all going in the same direction or that the network is a tree. \Cref{F:scn_unordered} gives another example, where the same raw material $L$ (labor) is used by three firms, $T$ (transport), $F$ (final goods producer), and $C$ (communication). These three firms set their prices independently, but $F$ additionally takes the markups of the $D$ (distributor) and $R$ (retailer) as given.
\begin {figure}[!ht]
\begin{minipage}{0.5\textwidth}
  \centering
  \scalebox{0.85}{
  \begin{tikzpicture}
  [
  cnode/.style={draw=black,fill=#1,minimum width=3mm,circle},
  ]
  \node[cnode=white] (nL)   at (0,-2) {$L$};

  \node[cnode=white] (nT)  at (2,-0.66){$T$};
  \node[cnode=white] (nF)  at (2,-2)   {$F$};
  \node[cnode=white] (nC)  at (2,-3.33){$C$};

  \node[cnode=white] (nD)   at (4,-2)  {$D$};
  \node[cnode=white] (nR)   at (6,-2)  {$R$};

  \draw[->] (nL) -- (nT);
  \draw[->] (nL) -- (nF);
  \draw[->] (nL) -- (nC);

  \draw[->] (nR) -- (nD);
  \draw[->] (nD) -- (nF);
  
  \draw[->,out=130,in=50]   (nR) edge (nF);
\end{tikzpicture}
  }
\end{minipage}
\begin{minipage}{0.5\textwidth}
  \[
  \madj = 
  \kbordermatrix{ 
    & L & T & F & C & D & R \\
    L & 0 & 1 & 1 & 1 & 0 & 0 \\ 
    T & 0 & 0 & 0 & 0 & 0 & 0 \\ 
    F & 0 & 0 & 0 & 0 & 0 & 0 \\ 
    C & 0 & 0 & 0 & 0 & 0 & 0 \\ 
    D & 0 & 0 & 1 & 0 & 0 & 0 \\ 
    R & 0 & 0 & 1 & 0 & 1 & 0 \\
  }
  \]
\end{minipage}
  \caption{Example: a network with a small producer and a common raw-material producer} \label{F:scn_unordered}
\end{figure}

\subsection{Regularity Assumptions}

I make three technical assumptions that are sufficient for the existence and uniqueness of the equilibrium. These assumptions are natural in most applications and simplify the analysis, but could be relaxed. The first assumption specifies the class of networks.
\begin{assumption} \label{A:ass_network}
  Network $\madj$ is acyclic and transitive.\footnote{Acyclicity: $\nexists i_1,\dots,i_k$ such that $a_{i_1 i_2}=\dots=a_{i_{k-1} i_k} = a_{i_k i_1} = 1$. Equivalently, $\madj^n = \mathbf{0}$. Transitivity: if $a_{ij}=a_{jk}=1$, then $a_{ik}=1$. 
  Equivalently, $\madj \geq \madj^2$.
  }
\end{assumption}
It is natural to assume that the network is acyclic. If firm $j$ takes $p_i$ as given, then it simply cannot be that firm $i$ takes its decision $p_j$ as given. Similarly for cycles of more than two players. Note the assumption allows firms $i$ and $j$ to make independent choices when $a_{ij}=a_{ji}=0$ and the network does not have to be connected.

The transitivity requires that if $i$ influences $j$ and $j$ influences $k$, then $i$ also influences $k$ directly, i.e.\ $k$ takes both $p_j$ and $p_i$ as given. In the examples above this was a natural assumption. Relaxing transitivity assumption would add the possibility of signaling to the game. For example, suppose that in the network described by \cref{F:scn_unordered} there is no edge from $R$ to $F$. Then firm $F$ knows that $p_R$ will be added to the price, but does not know the value. However, since $F$ knows $p_D$ and $D$ knows $p_R$, the price $p_D$ may reveal some information about $p_R$. Transitivity assumption excludes such signaling possibilities and thus simplifies the analysis significantly.\footnote{Compared to typical signaling models, the private information here is about the choices of other players (deviations in particular) rather than some underlying uncertainty.}

The second regularity assumption puts standard restrictions on the demand function. The demand function $D(P)$ is a smooth and strictly decreasing function. It either has a finite saturation point $\oP$ at which the demand is zero or converges to zero fast enough so that profit maximization problem is well-defined.

\begin{assumption} \label{A:ass_demand}
  Demand function $D : [0,\oP) \to \R_+$ is continuously differentiable and strictly decreasing in $[0,\oP)$ where $\oP \in \R_+$ could be finite or infinite.
  Moreover, it satisfies limit condition $\lim_{P \to \infty} P D(P) = 0$.
\end{assumption}

The third and final regularity assumption ensures that the demand function $D(P)$ is well-behaved so that the optimum of each firm can be found using the first-order condition. 
It is typical in the literature to make a regularity assumption that $D$ is twice differentiable and profits single-peaked. In particular, in theoretical works the demand is often assumed to be linear for tractability. However, in empirical literature logit demand is more common. Here I make an assumption about the demand function that would be analogous to the standard regularity assumption and covers both linear and logit demand functions.

Let the depth of the network $d(\madj)$ be the length of the longest path in $\madj$.\footnote{Formally, $d(\madj)$ is smallest $d$ is such that $\madj^{d} = \mathbf{0}$.} For example, in \cref{F:scn_unordered} depth $d(\madj) = 3$ (from the path $R \to D \to F$). Moreover, let me define a function
\begin{equation} \label{E:gdef}
  g(P) = - \frac{D(P)}{D'(P)},
\end{equation}
which is a convenient alternative way to represent the demand function. Note that $g(P) = \frac{P}{\varepsilon(P)}$, where $\varepsilon(P) = - \dd{D(P)}{P} \frac{P}{D(P)}$ is the demand elasticity. 
Then I make the following assumption about the shape of the demand function.
\begin{assumption} \label{A:ass_monotone}
  $g(P)$ is strictly decreasing and $d(\madj)$-times monotone in $P \in (0,\oP)$, i.e.\ for all $k=1,\dots,d(\madj)$, derivative $\dd[k]{g(P)}{P}$ exists and $(-1)^k \dd[k]{g(P)}{P} \geq 0$ for all $P \in (0,\oP)$.
\end{assumption}
To interpret the condition, let us look at the standard monopoly pricing problem $\max_P \pi(P) = \max_P (P-C) D(P)$. Then the first-order necessary condition for optimality of $P^*$ is
\begin{equation}
  \pi'(P^*) 
  = D(P^*) + (P^*-C) D'(P^*) = 0
  \;\;\Rightarrow\;\;
  P^* - C 
  = g(P^*),
\end{equation}
which illustrates the convenience of the $g(P)$ notation. Moreover, a sufficient condition for optimality is $\pi''(P^*) \leq 0$ or equivalently
$2 [D'(P^*)]^2 \geq D(P^*) D''(P^*)$. Note that a sufficient condition for this is $[D'(P^*)]^2 \geq D(P^*) D''(P^*)$, which is equivalent to $-g'(P^*) \geq 0$.
Therefore in the standard monopoly problem, monotonicity of $g(P)$ guarantees that monopoly profit has a unique maximum that can be found using the first-order approach. For general networks, the condition is stronger, it also guarantees that best-responses and best-responses to best-responses are well-behaved, so that the first-order approach is valid.

As illustrated by the monopoly example, the condition is sufficient and not necessary, but it is easy to check and it is satisfied for many applications. The following proposition provides a formal statement, by showing that with many typical functional form assumptions on $D(P)$, the function $g(P)$ is completely monotone, i.e.\ $d$-times monotone for arbitrarily large $d \in \N$. Therefore \cref{A:ass_monotone} is satisfied with all networks.

\begin{proposition}[Many demand functions imply completely monotone $g(P)$] 
  \label{E:completelymonotonefunctions}
  Each of the following demand functions implies $d$-times monotone $g(P)$ for any $d \in \N$:
  \begin{enumerate}
    \item Linear demand $D(P)= a-bP$ with $a,b>0$ $\Rightarrow g(P) = \oP - P$, where $\oP = \frac{a}{b}>0$.
    \item Power demand $D(P)=d \sqrt[\beta]{a-bP}$ with $d,\beta,a,b>0$ $\Rightarrow g(P) = \beta (\oP-P)$.
    \item Logit demand $D(P)=d \frac{e^{-\alpha P}}{1+e^{-\alpha P}}$ with $d,\alpha>0$ $\Rightarrow g(P) = \frac{1}{\alpha} \left[ 1+e^{-\alpha P} \right]$.
    \item Exponential demand $D(P) = a-b e^{\alpha P}$ with $a > b >0, \alpha>0 \Rightarrow g(P) = \frac{1}{\alpha} \left[
      \oP e^{-\alpha P} - 1
    \right]$.
  \end{enumerate}
\end{proposition}

Note that for all four functions \cref{A:ass_demand} is clearly also satisfied. Linear and power demand functions have saturation point $\oP$, logit demand satisfies $
\lim_{P \to \infty} P D(P) 
= d \lim_{P \to \infty} \frac{1}{\alpha e^{\alpha P}} = 0$, and exponential demand has saturation point $\oP = \frac{1}{\alpha} \log \frac{a}{b}$.

\section{Characterization} \label{S:characterization}

In this section, I first discuss two examples that illustrate the complications that arise from the non-linearities and the network structure. I use these examples to illustrate the techniques that I use for equilibrium characterization. The main result of the paper is the characterization theorem, that formalizes the approach.

\subsection{Example: Non-linear Demand}

The first example illustrates the complexities that arise from working with non-linear demand function, when the decisions are sequential. Let us consider logit demand $D(P) = \frac{e^{-P}}{1+e^{-P}}$, costless production, and two monopolists, who choose their prices sequentially, as indicated in \cref{F:2seq}. That is, first firm $1$ chooses price $p_1$, which is taken into account by firm $2$, when it chooses price $p_2$. The price of the final good is $P=p_1+p_2$.
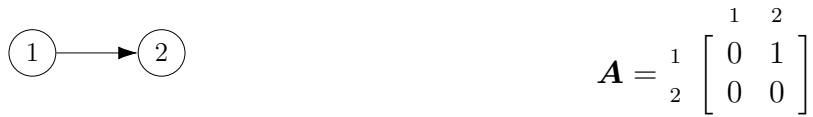
\begin{figure}[!ht] 
\begin{minipage}{0.5\textwidth}
  \centering
  \scalebox{0.85}{
  \begin{tikzpicture}
  [
  cnode/.style={draw=black,fill=#1,minimum width=3mm,circle},
  ]
  \node[cnode=white] (nF)   at (4,-2)   {$1$};
  \node[cnode=white] (nD)   at (6,-2)   {$2$};
  \draw[<-] (nD) -- (nF);
  \end{tikzpicture}  
  }
\end{minipage}
\begin{minipage}{0.5\textwidth}
  \[
  \madj = 
  \kbordermatrix{ 
      & 1 & 2 \\
    1 & 0 & 1 \\ 
    2 & 0 & 0 \\ 
  }
  \]
\end{minipage}  
  \caption{Example: two sequential monopolists} \label{F:2seq}
\end{figure}

The standard method of finding the equilibrium in this game is backward-induction It starts by finding the best-response function of firm $2$, by solving
$\max_{p_2} p_2 D(p_1+p_2)$. The optimality condition is
\begin{equation} \label{E:compl1foc2}
  \dd{\pi_2}{p_2}
  =
  D(p_1+p_2)
  + p_2 D'(p_1+p_2)
  = 0
\;\;\iff\;\;
  e^{p_2} (1-p_2) = e^{-p_1}.
\end{equation}
Solving this, gives best-response function $p_2^*(p_1) = 1 + W(e^{-(p_1+1)})$, where $W(\cdot)$ is the Lambert W function\footnote{Function $W(x)$ is defined as a solution to $x = W(x) e^{W(x)}$.}. Substituting $p_2^*(p_1)$ to the optimization problem of firm $1$ gives
\begin{equation} \label{E:compl1example}
  \max_{p_1} p_1 D(p_1+p_2^*(p_1))
  \Rightarrow
  1 + e^{-(p_1+1)-W(e^{-(p_1+1)})} = p_1 \left( 1+\frac{-e^{-(p_1+1)} W(e^{-(p_1+1)})}{e^{-(p_1+1)}+W(e^{-(p_1+1)})} \right).
\end{equation}
Solving it numerically gives $p_1^* \approx 1.2088$, therefore $p_2^* \approx 1.0994$ and $P^* \approx 2.3082$.
However, there is no analytic solution to optimality condition \eqref{E:compl1example}. This implies that the standard approach fails when the network is more complex than the one studied here. The backward-induction cannot be used, since computing best-response functions and substituting them to the maximization problems of other firms is not feasible. The issue is tractability---since the optimality conditions are non-linear, solving them leads to complex expressions. Replacing best-responses sequentially amplifies these complexities.
 
The solution to this problem comes from \cite{hinnosaar_optimal_2018}, which proposes characterizing the behavior of the following players by inverted best-response functions. The key observation is that although \cref{E:compl1foc2} is a highly non-linear function of $p_1$ and $p_2$ separately, fixing $P=p_1+p_2$ leads to a linear equation for $p_2$ or equivalently $p_1 = P-p_2$. 
Therefore, for a fixed price of the final good, $P$, it is straightforward to find the price if firm 2 that is consistent with the final good price $P$. 
I denote this by $f_2(P)$, i.e.
\[
  f_2(P) = p_2 = -\frac{D(P)}{D'(P)} = g(P) 
\]
where $g(P) = 1+e^{-P}$. 
Firm $1$ knows that if it sets its price to $p_1$, the price of the final good will satisfy $P = p_1 + f_2(P)$. Therefore we can think of firm 1's problem as choosing $P$ to solve
\[
  \max_{P} [P-f_2(P)] D(P)
  \;\;
  \Rightarrow
  \;\;
  f_1(P) = p_1 
  = P-f_2(P) 
  = g(P) [1-f_2'(P)].
\]
Therefore if the final good's price in equilibrium is $P^*$, then optimal behavior of both players requires that $P^* = f_1(P^*) + f_2(P^*) = 2 g(P^*) - g'(P^*) g(P^*)$. This equation is straightforward to solve and the same argument can be easily extended to more players choosing sequentially.

There is one more pattern in the particular expressions we get that the characterization will exploit. Namely, the condition for equilibrium is
\[
  P^* 
  = 2 g(P^*) - g'(P^*) g(P^*)
  = 2 g_1(P^*) + g_2(P^*),
\]
where $g_1(P^*)=g(P^*)$ and $g_2(P^*) = -g_1'(P^*) g(P^*)$.
The expression on the right-hand side consists of two elements. The first, $2 g_1(P^*)$ captures the fact that there are two players who each individually maximize their profits. The second $g_2(P^*)$ captures the fact that player $1$ influences player $2$. It is straightforward to verify that for example if we would remove this influence, i.e.\ with two monopolists choosing their prices simultaneously, the equilibrium condition would become $P^* = 2 g_1(P^*)$.

The main advantage of this approach is tractability. Instead of solving non-linear equations in each step and inserting the resulting expressions to the next maximization problems, which results in more complex non-linear expressions, this approach allows combining all necessary conditions of optimality into one necessary condition. Under \cref{A:ass_demand,A:ass_monotone}, the resulting expression has a unique solution, which gives us a unique candidate for an interior equilibrium. Under the same assumptions, the sufficient conditions for optimality are also satisfied and therefore it determines unique equilibrium.

\subsection{Example: Interconnected Decisions} \label{SS:interconnected}

The second example illustrates a new issue that arises in the case of networks---the decisions are interconnected. For example, on the network depicted by \cref{F:scn_unordered}, firms $L$ and $D$ make independent decisions, but due to their positions, they have different views on what happens before and after them. Firm $D$ influences only $F$, but $L$ influences $T$ and $C$ as well. Similarly, $D$ takes $p_R$ as given, whereas $L$ does not observe $p_R$ and therefore has to have an equilibrium conjecture about the optimal behavior of $R$. Therefore solving the game sequentially is not possible anymore.

To illustrate this issue, consider the simple network illustrated by \cref{F:ex_nonnested}. Let us assume that the demand is linear $D(P)=1-P$, and there are no costs and no price-takers.
\begin{figure}[!ht]
\begin{minipage}{0.5\textwidth}
  \centering   
  \scalebox{0.85}{
  \begin{tikzpicture}
  [
  cnode/.style={draw=black,fill=#1,minimum width=3mm,circle},
  ]
  \node[cnode=white] (n1) at (2,0) {1};
  \node[cnode=white] (n2) at (2,-2) {2};
  \node[cnode=white] (n3) at (4,0) {3};
  \node[cnode=white] (n4) at (4,-2) {4};
  \draw[->] (n1) -- (n3);
  \draw[->] (n1) -- (n4);
  \draw[->] (n2) -- (n4);
  \end{tikzpicture}%
  }
\end{minipage}
\begin{minipage}{0.5\textwidth}
  \[
  \madj = 
  \kbordermatrix{ 
      & 1 & 2 & 3 & 4 \\
    1 & 0 & 0 & 1 & 1 \\ 
    2 & 0 & 0 & 0 & 1 \\ 
    3 & 0 & 0 & 0 & 0 \\ 
    4 & 0 & 0 & 0 & 0 \\ 
  }
  \]
\end{minipage}    
  \caption{Example: network with interconnected decisions}\label{F:ex_nonnested}
\end{figure}
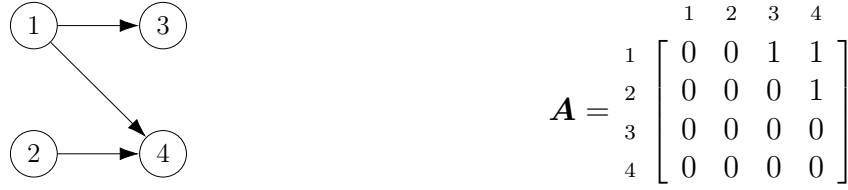

The strategies of the firms are respectively $p_1^*$, $p_2^*$, $p_3^*(p_1)$, and $p_4^*(p_1,p_2)$. Let us first consider the problem of player $4$, who observes $p_1$ and $p_2$ and expects equilibrium behavior from player $3$. Therefore player $4$ solves
\[
\max_{p_4 \geq 0} p_4 \left[1-p_1-p_2-p_3^*(p_1)-p_4\right]
\;\;\;
\Rightarrow
\;\;\;
p_4^*(p_1,p_2) = \frac{1}{2} \left[1-p_1-p_2-p_3^*(p_1) \right].
\]
While this condition provides a condition for the best-response function $p_4^*(p_1,p_2)$, we have not yet characterized it, as it would require knowing $p_3^*(p_1)$. Player $3$ solves a similar problem, but does not observe $p_2$ and expects $p_4$ to be $p_4^*(p_1,p_2^*)$
\[
\max_{p_3 \geq 0} p_3 \left[1-p_1-p_2^*-p_3-p_4^*(p_1,p_2^*)\right]
\;\;\;
\Rightarrow
\;\;\;
p_3^*(p_1) = \frac{1}{2} \left[1-p_1-p_2^*-p_4^*(p_1,p_2^*(p_1)) \right].
\]
Again, computing this best-response function explicitly, requires knowing $p_4^*(p_1,p_2)$, but also the equilibrium price of player $2$, i.e.\ $p_2^*$. 
To compute the best-response functions explicitly (i.e.\ independently of each other), we first need to solve the equation system that we get by inserting $p_2^*$ to the optimality condition of player $4$. This gives us
\[
p_3^*(p_1) = p_4^*(p_1,p_2^*) = \frac{1}{3} \left[1-p_1-p_2^*\right]
\;\;\Rightarrow \;\;
p_4^*(p_1,p_2)
= \frac{1}{3}\left[1-p_1\right]+\frac{1}{6 }p_2^*-\frac{1}{2} p_2.
\]
Note the prices $p_3$ and $p_4$ we have now characterized are still not the true best-response functions, since they depend on the equilibrium price $p_2^*$, which is yet to be determined. For this we need to solve the problem of player $2$, who expects player $1$ to choose equilibrium price $p_1^*$
\[
\max_{p_2 \geq 0} p_2 \left[1-p_1^*-p_2-p_3^*(p_1^*)-p_4^*(p_1^*,p_2) \right].
\]
Taking the first-order condition and evaluating it at $p_2=p_2^*$ gives a condition
\begin{equation} \label{E:ex_nonnested_foc2}
\frac{1}{6}\left[ 2-2 p_1^*-5 p_2^* \right] = 0.
\end{equation}
Finally, player 1 solves a similar problem, taking $p_2^*$ as fixed, i.e.
\[
\max_{p_1 \geq 0} p_1 \left[1-p_1-p_2^*-p_3^*(p_1)-p_4^*(p_1,p_2^*) \right].
\]
Again, taking the first-order condition and evaluating it at $p_1=p_1^*$ gives
\begin{equation} \label{E:ex_nonnested_foc1}
\frac{1}{3}\left[1-2 p_1^*-p_2^*\right] = 0.
\end{equation}
Solving the equation system \cref{E:ex_nonnested_foc1,E:ex_nonnested_foc2} gives us $p_1^*= \frac{3}{8}$, $p_2^* = \frac{1}{4}$. Inserting these values to the functions derived above gives us the best-response functions $p_3^*(p_1) = \frac{1}{4}-\frac{1}{3}p_1$ and $p_4^*(p_1,p_2)=\frac{3}{8}-\frac{1}{3}p_1-\frac{1}{2}p_2$. We can also compute the equilibrium prices $p_3^*(p_1^*) = p_4^*(p_1^*,p_2^*) = \frac{1}{8}$. Therefore equilibrium price of the final good is $P^* = \frac{7}{8}$.

As the example illustrates, finding the equilibrium strategies requires solving a combination of equation systems in parallel with finding the best-response functions. Each additional edge in the network can create a new layer of complexity.

The inverted best-response approach solves this issue as follows. Consider the optimization problem of firm $4$. For given $(p_1,p_2)$, it chooses optimal $p_4$. We can rethink its optimization problem as choosing the final good price $P = p_1 + p_2 + p_3^*(p_1) + p_4$ that it wants to induce. 
We can rewrite its maximization problem as
\[
  \max_{P} [P-p_1-p_2-p_3^*(p_1)] D(P)
  \;\;\Rightarrow\;\;
  D(P)
  +
  [P-p_1-p_2-p_3^*(p_1)] D'(P) 
  = 0,
\]
which can be rewritten as $P-p_1-p_2-p_3^*(p_1) = -\frac{D(P)}{D'(P)} = g(P) = 1-P$. 
This is a necessary condition for optimality, but since the problem is quadratic, it is easy to see that it is also sufficient. This expression gives implicitly the best-response function $p_4^*(p_1,p_2)$. But more directly, the expression in the left-hand-side is the optimal $p_4$ that is consistent with the final good price $P$ and the optimal behavior of firm $4$. Let us denote it by $f_4(P) = 1-P$. The problem for the firm $3$ is analogous and gives $f_3(P)=1-P$.

Now, consider firm $2$. Instead of choosing $p_2$ it can again consider the choice of the final good price $P$. Since only firm $4$ observes its choice (and thus chooses $p_4=f_4(P)$ as a response to desired $P$), the firm $2$'s problem can be written as
\[
  \max_{P} [P-p_1^*-p_3^*(p_1^*)-f_4(P)] D(P)
  \;\;\Rightarrow\;\;
  [1-f_4'(P)] D(P)
  + [P-p_1^*-p_3^*(p_1^*)-f_4(P)] D'(P)
  =0,
\]
which gives us a condition $f_2(P)=P-p_1^*-p_3^*(p_1^*)-f_4(P) = 2 (1-P)$. Analogous calculation for firm $1$ gives $f_1(P) = 3(1-P)$.
Now, the equilibrium price $P^*$ of the final good must be consistent with individual choices. Therefore we get a condition
\[
  P^* = \sum_{i=1}^4 f_i(P^*) = 7(1-P^*).
\]
Solving this equation gives us $P^* = \frac{7}{8}$ and individual prices $p_1^* = 3(1-P^*) = \frac{3}{8}$, $p_2^* = \frac{2}{8}$, and $p_3^* = p_4^* = \frac{1}{8}$.

Notice that the same calculations could be applied easily for non-linear demand functions, with some $g(P) = -\frac{D(P)}{D(P)}$. This would give us an equilibrium condition
\[
  P^* 
  = \sum_{i=1}^4 f_i(P^*) 
  = 4 g_1(P^*) + 3 g_2(P^*),
\]
where $g_1(P)=g(P)$ and $g_2(P)=-g_1'(P) g(P)$. This is again the same pattern that we saw in the previous example, since the number of players is $4$ and the number of edges is $3$. In the case of linear demand, $g_1(P) = g(P) = 1-P$ and therefore $g_2(P) = 1-P$.

This example illustrates the advantage of the inverted best-response approach. As the approach combines all necessary conditions into one, the issues of interconnected decisions are automatically mitigated.

\subsection{Characterization}

As illustrated by the examples above, it is useful to define functions $g_1,\dots,g_n$, which capture relevant properties of the demand function. They are defined recursively as
\begin{equation} \label{E:gidef}
g_1(P) = g(P) = -\frac{D(P)}{D'(P)}
\text{ and }
g_{k+1}(P) = -g_k'(P) g(P).
\end{equation}
As the discussion about monopoly profit maximization and the examples illustrated, $g_1(P)$ captures the standard concavity of the profit function, whereas $g_2(P)$ captures the direct discouragement effect when a firm observes the price of another firm. Functions $g_3,\dots,g_n$ play a similar role in describing higher-order discouragement effects.

Note also that the adjacency matrix $\madj$ provides a convenient way to keep track of the number of direct and indirect influences. Multiplying adjacency matrix with a column vector of ones, $\madj \vone$, gives a vector with the number of edges going out from each player (i.e.\ the sum over columns). Similarly, $\vone' \madj \vone$ is the total number of edges on the network, i.e.\ the total number of direct influences. Multiplying the adjacency matrix by itself, i.e.\ $\madj^2 = \madj \madj$ gives a matrix that describes two-edge paths, i.e.\ element $a_{i,j}^2$ is the number of paths from $i$ to $j$ with one intermediate step. Similarly $\madj^k$ is the matrix that describes number of all $k$-step paths from each $i$ to each $j$. When we take $k=0$, then $\madj^0$ is an identity matrix, which can be interpreted as $0$-step paths (clearly the only player that can be reached from player $i$ by following $0$ edges is player $i$ himself). 

Therefore $\madj^k \vone$ is a vector whose elements are the numbers of $k$-step paths from player $i$, which can be directly computed as $\ve_i' \madj^k \vone$, where $\ve_i$ is a column vector, where $i$th element is $1$ and other elements are zeros. Similarly, $\vone' \madj^k \vone$ is the number of all $k$-step paths in the network. 
The following expression makes these calculations for the network described by \cref{F:scn_unordered}, which has six players, six edges, and one two-edge path ($R \to D \to F$).
\begin{equation} \label{E:scn_unordered}
\kbordermatrix{
  & \madj^0 \vone & \madj^1 \vone & \madj^2 \vone & \madj^3 \vone \\
  L & 1 & 3 & 0 & 0 \\ 
  T & 1 & 0 & 0 & 0 \\ 
  F & 1 & 0 & 0 & 0 \\ 
  C & 1 & 0 & 0 & 0 \\ 
  D & 1 & 1 & 0 & 0 \\ 
  R & 1 & 2 & 1 & 0 \\
  \vone' \madj^{k-1} \vone & 6 & 6 & 1 & 0 \\
}.
\end{equation}
%

With this notation, I can now state the main result of this paper, the characterization theorem that states that there exists a unique equilibrium and shows how it is characterized using the components we have discussed.

\begin{theorem}\label{T:characterization}
  There is a unique equilibrium, the final good price $P^*$ is the solution to 
  \begin{equation} \label{E:equilibrium}
    P^* - C  
    = \sum_{k=1}^n \vone' \madj^{k-1} \vone g_k(P^*)
    ,
  \end{equation}
  and the individual prices are $p_i^* = c_i 
  + \sum_{k=1}^n \ve_i' \madj^{k-1} \vone g_k(P^*)
  $ for all $i$.  
\end{theorem}
The proof in \cref{A:proofs} builds on the ideas discussed above. A few remarks are in order. The uniqueness is straightforward to establish. \Cref{A:ass_monotone} implies that each $g_k(P)$ is weakly decreasing (this is formally shown in \cref{L:gkmonotone} in \cref{A:proofs}). The right-hand side of \cref{E:equilibrium} is therefore decreasing, whereas the left-hand side is strictly increasing.
Connection to inverted best-response functions is also clear, as the individual prices are determined by $p_i^* = c_i + f_i(P^*)$.


\section{Multiple-marginalization Problem} \label{S:marginalization}

Let me first interpret the equilibrium condition \cref{E:equilibrium} by comparing it with the known benchmark cases.
First, when all firms are price-takers, then the network is empty and therefore the right-hand side of \eqref{E:equilibrium} is zero. As expected, the equilibrium condition is, therefore, $P^*=C$, i.e.\ price of the final good equals the marginal cost of the final good. Standard arguments imply that this is also the welfare-maximizing solution.

Second, suppose that there is a single monopolist, i.e.\ $n=1$ and $\madj = [0]$. Therefore there is a single element on the right-hand side of \eqref{E:equilibrium} with value $g_1(P^*)$. We can rewrite the condition as
\begin{equation}
  \frac{P^*-C}{P^*}
  = \frac{g_1(P^*)}{P^*}
  = \frac{1}{\varepsilon(P^*)},
\end{equation}
which is the standard inverse-elasticity rule: mark-up (Lerner index) equals the inverse elasticity. There is a usual monopoly distortion---as the monopolist does not internalize the impact on the consumer surplus, the price equilibrium price of the final good is higher and the equilibrium quantity lower than the social optimum.
This is also the joint profit-maximization outcome.

Third, consider $n > 1$ monopolists who are making their decisions simultaneously. That is, the network has $n$ nodes, but no edges. Analogously with a single monopolist, we can then rewrite the equilibrium condition as
\begin{equation}
  \frac{P^*-C}{P^*} 
  = \vone' \madj^0 \vone \frac{g_1(P^*)}{P^*}
  = \frac{n}{\varepsilon(P^*)} > \frac{1}{\varepsilon(P^*)}.
\end{equation}
The total markup is now strictly higher than in the case of a single monopolist. This is the standard \emph{multiple-marginalization problem}---firms do not internalize not only the impact on consumer surplus but also the impact on the other firms. Therefore the distortion is even larger than in the case of a single monopolist, which means that both the total profits and the social welfare are reduced compared to a single monopolist.

The novel case studied in this paper is with multiple monopolists and some influences. That is $n>0$ and $\madj \neq \vzero$. In this case, the condition can be written as
\begin{equation}
  \frac{P^*-C}{P^*}
  = \frac{n}{\varepsilon(P^*)} + \sum_{k=2}^n \vone' \madj^{k-1} \vone \frac{g_k(P^*)}{P^*} > \frac{n}{\varepsilon(P^*)}.
\end{equation}
The total markup and therefore the distortion is even higher than with $n$ independent monopolists. The intuition for this is simple: suppose that there is a single edge so that firm $i$ influences firm $j$. Then in addition to the trade-offs firm $i$ had before, raising the price now will reduce the profitability of firm $j$, who will respond by reducing its price. Therefore $p_i$ will be higher and $p_j$ lower than with simultaneous decisions. How about the price of the final good, which depends on the sum of $p_i$ and $p_j$? If the reduction in $p_j$ would be so large that the total price does not increase, then $p_i$ would not be optimal, since the profit of firm $i$ is $(p_i-c_i) D(P)$, i.e.\ increasing in $p_i$ and decreasing $P$, so $i$ would want to raise the price even further. Thus in equilibrium, it should be that the price of the final good is increased.
I formalize and generalize this observation as \cref{C:multimarg}. The corollary follows from \cref{E:equilibrium} and non-negativity of $g_k$ functions.

\begin{corollary}[Magnified Multiple-marginalization Problem] \label{C:multimarg}
  Suppose that there are two networks $\madj$ and $\madjother$ such that
  \begin{enumerate}
    \item $\vone' \madj^{k-1} \vone \geq \vone' \madjother^{k-1} \vone$ for all $k \in \{1,\dots,n\}$ and
    \item $\vone' \madj^{k-1} \vone > \vone' \madjother^{k-1} \vone$ for at least one $k$,
  \end{enumerate}
  then both social welfare and total profit in the case of $\madj$ is lower than with $\madjother$.
\end{corollary}

The result shows that multiple-marginalization problem is increased with strategic influences, but does not give a magnitude for it. To illustrate that the impact may be severe, let me give some numeric examples. First, suppose that demand is linear, $D(P)=1-P$, there are no costs, and there are no price-takers. Then standard calculations imply that the maximized total welfare would be $\frac{1}{2}$ and a single monopolist would choose price $\frac{1}{2}$, which would lead to dead-weight loss of $\frac{1}{8}$. Therefore with any network dead-weight loss is at least $\frac{1}{8}$ and at most $\frac{1}{2}$. 
\Cref{F:dwl_comp} illustrates the difference between the dead-weight loss in the best case (simultaneous decisions) and the worst case (sequential decisions).
Even in the best case (blue line with triangles) the multiple-marginalization problem can be severe and is increasing in $n$. However, the distortions with strategic interactions (red line with circles) are much higher for any $n$ and the dead-weight loss approaches to full destruction of the social welfare quickly. This comparison shows that strategic influences magnify the multiple-marginalization problem with any $n$.
\begin{figure}[ht!]
  \centering%
  \includegraphics[trim={0.22cm 0.45cm 0.86cm 0.35cm},clip,width=0.5\linewidth]{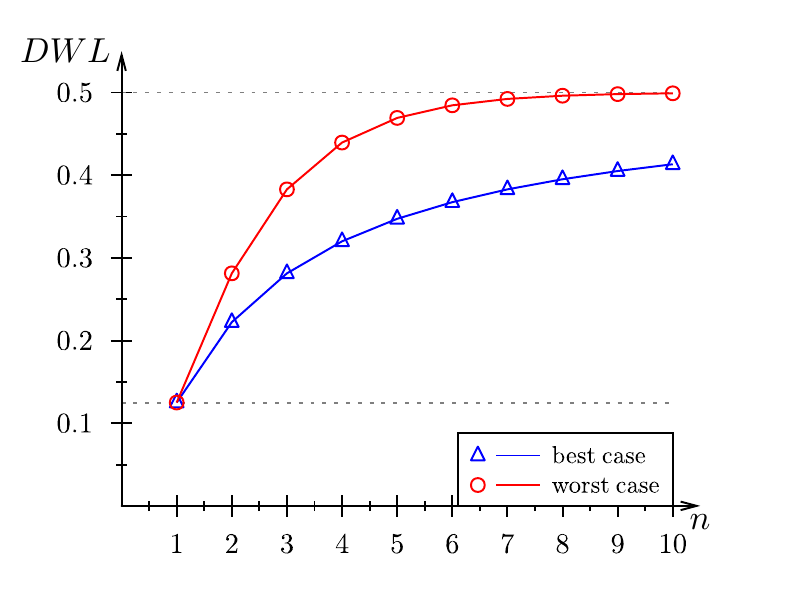}
  \caption{Example: comparison of dead-weight loss in the model with linear demand between the best case (simultaneous decisions) and the worst case (sequential decisions)} \label{F:dwl_comp}
\end{figure} 

How much the number of firms matters compared to strategic influences depends on the shape of the demand function. This is illustrated by \cref{F:dwl_comp_beta}, which provides the same comparison with a more general power demand function $D(P)=\sqrt[\beta]{1-P}$.\footnote{
Generally $DWL(P^*) = \int_0^{P^*} [D(P)-D(P^*)] dP$.  
The calculations and the impact of $\beta$ in case of power demand function are discussed in more detail in the next section.} When $\beta < 1$, then relatively high weight is given to more direct influences. Indeed, \cref{F:dwl_comp_beta01} shows that when $\beta=\frac{1}{10}$ then the best-case and the worst-case dead-weight loss do not differ much. On the other hand, when $\beta>1$ the weight is larger on more indirect influences. This is illustrated by \cref{F:dwl_comp_beta10}, where the difference between the two cases is large.
\begin{figure}[!ht]%
    \centering
    \subfloat[$\beta = \frac{1}{10}$]{\includegraphics[trim={0.22cm 0.35cm 0.86cm 0.35cm},clip,width=0.45\linewidth]{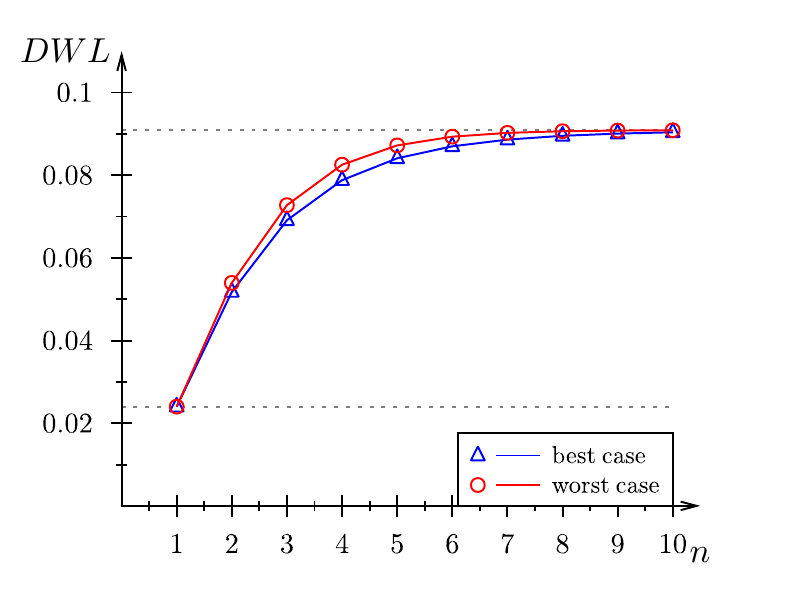} \label{F:dwl_comp_beta01}}
    \qquad
    \subfloat[$\beta = 10$]{\includegraphics[trim={0.22cm 0.35cm 0.86cm 0.35cm},clip,width=0.45\linewidth]{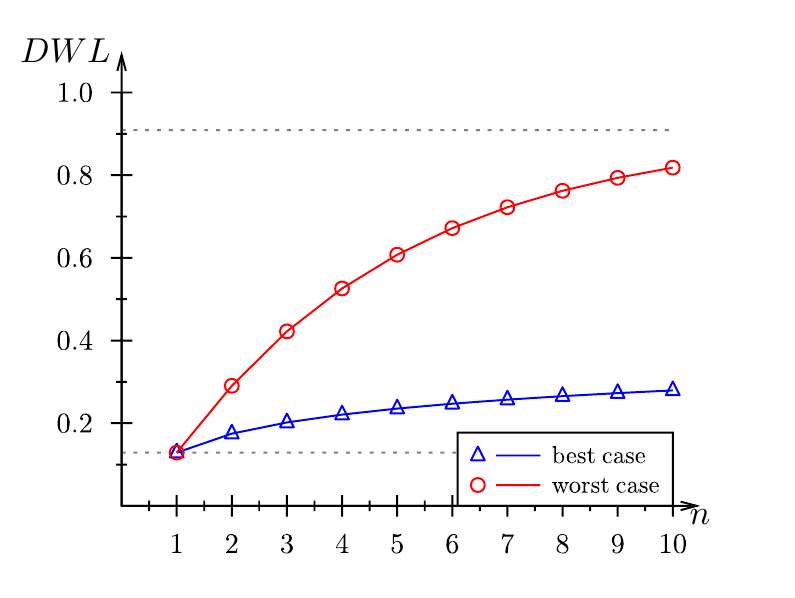} \label{F:dwl_comp_beta10}}
    \caption{Examples: comparison of dead-weight loss with power demand $D(P)=\sqrt[\beta]{1-P}$ between the best case (simultaneous decisions) and the worst case (sequential decisions) \label{F:dwl_comp_beta}}%
\end{figure}

\section{Influentiality} \label{S:influentiality}

\subsection{A Measure of Influentiality}

All monopolists on the network have some market power and therefore earn strictly positive profits. But some firms are more influential than others. Which ones and how does this depend on the network? 
The answer comes directly from the characterization in \cref{T:characterization}. For brevity, let me denote
\begin{equation} \label{E:influentialitydef}
  \infl_i(\madj) 
  = \sum_{k=1}^n \ve_i' \madj^{k-1} \vone g_k(P^*),
\end{equation}
which is a sum of scalars $\ve_i' \madj^{k-1} \vone$ weighted by $g_k(P^*)$. Note that $\ve_i' \madj^0 \vone = 1$, $\ve_i' \madj^1 \vone$ is the number of players $i$ influences, $\ve_i' \madj^2 \vone$ is the number of two-edge paths starting from $i$, and so on. Therefore $\infl_i(\madj)$ can be interpreted as a measure of influentiality of player $i$.

Fixing the equilibrium price of the final good $P^*$, the individual markups are $p_i^* - c_i = \infl_i(\madj)$ and therefore profits $\pi_i(\bp^*) = (p_i^*-c_i) D(P^*) = \infl_i(\madj) D(P^*)$.
Therefore $\infl_i(\madj)$ fully captures the details of the network that affect firm $i$'s action and payoff. \Cref{C:influentiality} provides a formal statement.

\begin{corollary}[$\infl_i(\madj)$ Summarizes Influences] \label{C:influentiality}
  $\infl_i(P^*) > \infl_j(P^*)$ if and only if $\pi_i(\bp^*) > \pi_j(\bp^*)$ and $p_i-c_i > p_j^*-c_j$.
\end{corollary}

This measure of influentiality $\infl_i(\madj)$ depends both on the network structure and the demand function. There are some cases when we can say more. In particular, if firm $i$ has more influences in all levels than firm $j$, i.e.
$\ve_i' \madj^{k-1} \vone \geq \ve_j' \madj^{k-1} \vone$ for all $k$ and the inequality is strict for at least one $k$, then $I_i(\madj) \geq I_j(\madj)$ regardless of the weights $g_k(P^*)$. The inequality is strict whenever $g_k(P^*)>0$ for $k$ such that $\ve_i' \madj^{k-1} \vone > \ve_j' \madj^{k-1} \vone$. 
For example, when firm $i$ influences firm $j$, then $\ve_i' \madj^{k-1} \vone \geq \ve_j' \madj^{k-1} \vone$ and the inequality is strict for at least $k=1$, so $I_i(\madj) > I_j(\madj)$ with any demand function.

\subsection{Connections with Network Centrality Measures}

The measure of influentiality defined above is reminiscent of the classic measures of centrality as the capture the same effects: a player is more influential if it influences either more players or more influential players. The difference is that while the classic centrality measures are defined purely using network characteristics, the influentiality measure defined here has endogenous weights that are determined by the model parameters such as the demand function, costs, and also by the price of the final good.

In some special cases, the connection is even closer. Consider the case of power demand $D(P) = d \sqrt[\beta]{a-bP}$. As discussed in \cref{SS:power}, it implies linear $g_k(P) = \beta^k (\oP - P)$. Therefore $\infl_i(\madj) = (\oP-P^*) \bonacich_i^{\beta}(\madj)$, where $\bonacich_i(\madj)^{\beta} = \sum_{k=1}^n \beta^k \ve_i' \madj^{k-1} \vone$ is the Bonacich centrality measure of player $i$. The general measure $I_i(\madj)$ can be thought as a generalization of Bonacich centrality where the weights are endogenously determined by the demand function and the equilibrium, rather than having exponential decay $\beta^k$.\footnote{Note that the standard definition of Bonacich centrality requires $\beta<1$, because otherwise the sum may not converge and thus the measure would not be well-defined. As the influence networks here are acyclic, any $\beta >0 $ is allowed. Values $\beta \geq 1$ arise whenever $D(P)=d \sqrt[\beta]{a-bP}$, which is a quite natural assumption for a demand function.}
Linear demand $D(P) = a-bP$ is a special case of power demand with $\beta=1$. Therefore the influentiality measure $I_i(\madj)$ simplifies to Bonacich centrality measure with $\beta=1$, i.e.\ equal weight for each level of influences.

However, the influence measure does not always have to have a flavor of Bonacich centrality. Let me provide two more examples to show this. First, suppose $D(P) = d e^{\sqrt{2(a-b P)}/b}$. This a specifically constructed demand function, which implies $g(P) = g_1(P) = \sqrt{2(a-b P)}$ and therefore $g_2(P) = b$, which means that $g_k(P) = 0$ for all $k>2$. With these weights the influentiality measure simplifies to $I_i(\madj) = \sqrt{2(a-b P^*)} + b \ve_i' \madj \vone$, i.e.\ depends only on the number of players directly influenced by player $i$. That is, the influentiality measure is a linear function of the degree centrality in this case.

For another example, consider logit demand $D(P) = d \frac{e^{-\alpha P}}{1+e^{-\alpha P}}$. As I will show in \cref{SS:logit}, it may lead to complex expressions, but when $n$ is large enough, then $g_1(P^*) \approx \frac{1}{\alpha}$ and $g_k(P^*) \approx 0$ for $k>1$. Therefore $\infl_i(\madj) \approx \frac{1}{\alpha}$. This means that in the case of logit demand with sufficiently many players, the network structure does not affect the pricing of the individual firms. The relevant centrality measure is approximately a constant.
These observations are summarized by \cref{T:connections_centrality_measures}.

\begin{table}[!ht]
    \centering
\begin{tabular}{l|l|l}
    Demand $D(P)$ & Influentiality $\infl_i(\madj)$ & Equivalent Network Centrality Measure \\
    \hline
    Power $d \sqrt[\beta]{a-b P}$  
    & $\left(\frac{a}{b}-P^* \right) \bonacich_i^{\beta}(\madj)$ 
    & $\bonacich_i^{\beta}(\madj)=$ Bonacich centrality (with $\beta$) \\ 
    Linear $a-b P$
    & $\left(\frac{a}{b}-P^* \right) \bonacich_i(\madj)$ 
    & $\bonacich_i(\madj)=$ Bonacich centrality with $\beta=1$ \\ 
    $d e^{\sqrt{2(a-b P)}/b}$
    & $\sqrt{2(a-b P^*)} + b D_i(\madj)$
    & $D_i(\madj) = \ve_i'\madj \vone =$ Degree centrality \\ 
    Logit $d \frac{e^{-\alpha P}}{1+e^{-\alpha P}}$
    & $\to \frac{1}{\alpha}$ & Approximately a constant
  \end{tabular}
    \caption{Examples of demand functions with which the measure of influentiality $I_i(\madj)$ simplifies to one of the standard network centrality measures}
    \label{T:connections_centrality_measures}
\end{table}

\section{Computing the Equilibrium} \label{S:computing}

In this subsection I show how the equilibrium characterization could be used to compute the equilibrium and study some of the most common demand functions where the characterization is even simpler.

\subsection{Linear Demand} \label{SS:linear}

Suppose that the demand function is linear $D(P)=a-bP$. Then $g(P) = -\frac{D(P)}{D'(P)} = \oP-P = g_1(P)$ with $\oP = \frac{a}{b}$ and therefore for all $k>1$, $g_{k+1}(P) = -g_k'(P) g(P) = \oP-P$. \Cref{E:equilibrium} simplifies to
\begin{equation} \label{E:equilibrium_linear}
  P^* - C 
  = \sum_{k=1}^n \vone' \madj^{k-1} \vone g_k(P^*)
  = (\oP-P^*) \bonacich(\madj),
\end{equation}
where $\bonacich(\madj) = \sum_{k=1}^n \vone' \madj^{k-1} \vone$ is the sum of the number of influences of all levels, i.e.\ the number of players ($\vone' \madj^0 \vone = n$) plus the number of edges plus the number of two-edge paths, and so on. \Cref{E:equilibrium_linear} is a linear equation and its solution is the equilibrium price
\begin{equation}
  P^*
  =
  \frac{C + \oP \bonacich(\madj)}{1+\bonacich(\madj)}.
\end{equation}
As we would expect, increasing costs and increasing demand ($\oP = \frac{a}{b}$ in particular) will raise the equilibrium price, but the pass-through is imperfect.
Increasing the number of firms or the number of connections between firms increases the equilibrium price through the marginalization effects discussed above.
Similarly, we can compute the markups for individual firms,
\begin{equation}
  p_i^* 
  = c_i+\sum_{k=1}^n \ve_i' \madj^{k-1} \vone g_k(P^*)
  = c_i+
  \frac{\bonacich_i(\madj)}{1+\bonacich(\madj)} (\oP-C)
  ,
\end{equation}
where $\bonacich_i(\madj) = \sum_{k=1}^n \ve_i' \madj^{k-1} \vone$ is the sum of influences of firm $i$, i.e.\ $\ve_i' \madj^0 \vone = 1$ (``influencing'' oneself) plus $\ve_i' \madj^1 \vone = $ number of players $i$ influences plus the number of paths starting from $i$. By construction $\bonacich(\madj) = \sum_{i=1}^n \bonacich_i(\madj)$. 
 
Consider the example of network described by \cref{F:scn_unordered}, for which the corresponding $\madj^{k-1}\vone$ terms are computed in \cref{E:scn_unordered}.
Suppose that $D(P)=1-P$, and there are no costs, and no price-takers ($C=0$). Then $\bonacich(\madj) = 6+6+1=13$ and therefore $P^* = \frac{\bonacich(\madj)}{1+\bonacich(\madj)} = \frac{13}{14}$. Similarly, individual prices $p_i^* = \frac{\bonacich_i(\madj)}{1+\bonacich(\madj)}$. For example, $p_L^* = \frac{4}{14}, p_T^* = p_F^* = p_C^* = \frac{1}{14}, p_D^* = \frac{2}{14}$, and $p_R^* = \frac{4}{14}$. 
In particular, observe that $p_L^* = p_R^*$, but for different reasons---firm $L$ influences three firms directly, whereas $R$ influences two firms directly and one indirectly. In the case of linear demand, these two types of influences are weighted equally.

\subsection{Power Demand} \label{SS:power}

The calculations are similar for more general power demand $D(P) = d\sqrt[\beta]{a-bP}$. Then $g(P) = \beta (\oP-P)$ with $\oP=\frac{a}{b}$ and therefore $g_k(P) = \beta^k (\oP-P)$, so that \cref{E:equilibrium} gives the same expression for the equilibrium price of the final good
\begin{equation}
  P^*
  =
  \frac{C+\oP \bonacich^{\beta}(\madj)}{1+\bonacich^{\beta}(\madj)}
\end{equation}
but now $\bonacich^{\beta}(\madj) = \sum_{k=1}^n \beta^k \vone' \madj^{k-1} \vone$, i.e.\ the influences in various levels are weighted by $1,\beta,\beta^2,\dots$. Then $\beta$ can be interpreted as a decay or discount factor for more indirect influences.\footnote{Note that $\beta > 0$ (as otherwise demand would not be decreasing), but it can be bigger or smaller than $1$. In fact, when $\beta=1$ the demand function is linear, so that $\bonacich^{\beta}(\madj) = \bonacich(\madj)$.}  Similarly for individual firms,
\begin{equation}
  p_i^*  
  = c_i+
  \frac{\bonacich_i^{\beta}(\madj)}{1+\bonacich^{\beta}(\madj)} (\oP-C)
,
\end{equation}
where $\bonacich_i^{\beta}(\madj) = \sum_{k=1}^n \beta^k \ve_i' \madj^{k-1} \vone$, i.e.\ influences are again weighted by factor $\beta^k$.

Consider the example from \cref{F:scn_unordered} again, with $C=0$, and demand function $D(P)=\sqrt{\beta}(1-P)$. 
In particular, if $D(P)=(1-P)^2$, then $\beta=\frac{1}{2}$ and therefore $\bonacich^{\frac{1}{2}}(\madj) = 6 + \frac{1}{2} 6 + \frac{1}{4} 1 = \frac{37}{4}$, so $P^* = \frac{37}{41}$. As anticipated above, since higher weight is on direct influences than indirect influences, firm $L$ sets a higher price (and earns higher profit) than firm $R$, $p_L^* = \frac{10}{41} > p_R^*=\frac{9}{41}$. This is also the reason why the difference in the worst-case and the best-case on \cref{F:dwl_comp_beta01} was relatively small.
On the other hand, if $D(P) = \sqrt{1-P}$, then $\beta = 2$ which implies $P^* = \frac{22}{23}$ and $p_L^* = \frac{7}{23} < p_D^* = \frac{9}{23}$, since now the weight is higher on indirect influences. This explains the larger difference in \cref{F:dwl_comp_beta10}.

\subsection{Logit Demand} \label{SS:logit}

Take logit demand $D(P) = d \frac{e^{-\alpha P}}{1+e^{-\alpha P}} 
$ with $\alpha > 0$. Then $g(P) = \frac{1}{\alpha} \left[ 1+ e^{-\alpha P} \right]$. 

Let us first consider the example discussed in previous subsections to illustrate how the characterization result could be used for more complicated demand functions. Suppose again that $C=0$ and the network is the one described by \cref{F:scn_unordered}. Since the depth of the network is $d(\madj)=3$, we need to compute functions
\begin{align*}
  g_1(P) &= g(P) = \frac{1}{\alpha} \left[ 1+ e^{-\alpha P} \right], \\
  g_2(P) &= -g_1'(P) g(P) 
  = \frac{1}{\alpha} \left[ 1+ e^{-\alpha P} \right] e^{-\alpha P}, \\
  g_3(P) &= -g_2'(P) g(P) 
  =
  \frac{1}{\alpha} \left[ 1+ e^{-\alpha P} \right]
   e^{-\alpha P} \left[ 1+2e^{-\alpha P} \right].
\end{align*}
The equilibrium condition \eqref{E:equilibrium} takes the form
$P^* = 6 g_1(P^*) + 6 g_2(P^*) + g_3(P^*)$, which is straightforward to solve numerically. For example, when $\alpha = 1$, we get 
\[
  P^*
  =
 6 + 13 e^{-P^*} + 9 e^{-2P^*} + 2 e^{-3P^*},
\]
which implies $P^* \approx 6.0313$ and individual prices $p_L^* \approx 1.0096, p_T^* = p_F^* = p_C^* \approx 1.0024, p_D^* \approx 1.0048$, and $p_R^* \approx 1.0096$.

The numeric results point to a more specific property of the equilibrium behavior in the case of logit demand. Namely, all prices are only slightly above $1$. Inspecting $g_k(P)$ functions above reveals the reason. Namely, the term $e^{-\alpha P^*}$ converges to zero as $P^*$ increases. Therefore, for sufficiently large $P^*$, the weight $g_1(P^*)$ converges to a constant $\frac{1}{\alpha}$, whereas the weights $g_k(P^*)$ for $k>1$ converge to zero. Therefore, if the equilibrium price $P^*$ is large enough, it is almost fully driven by the number of players. This observation is formalized as the following \cref{L:logit}.

\begin{lemma}[Approximate Equilibrium with Logit Demand] \label{L:logit}
  With logit demand $D(P)=d \frac{e^{-\alpha}}{1+e^{-\alpha P}}$, the price of final good $P^*$ and individual prices $p_i^*$ satisfy the following conditions
  \begin{enumerate}
      \item $P^* > C + \frac{n}{\alpha}$ and $p_i^* > c_i + \frac{1}{\alpha}$ for all $i$,
      \item $P^* = C + \frac{n}{\alpha} + O\left( \left[\frac{2}{e} \right]^n \right)$
      and $p_i^* = c_i + \frac{1}{\alpha} + O \left( \left[ \frac{2}{e} \right]^n \right)$ for all $i$.\footnote{
    Where $f(n) = O(g(n))$ means that $\limsup _{n \to \infty} \left| \frac{f(x)}{g(x)} \right|<\infty$.
  }
  \end{enumerate}
\end{lemma}

\Cref{L:logit} implies when $n$ is large enough, then $P^* \approx C + \frac{n}{\alpha}$ and each $p_i^* \approx c_i + \frac{1}{\alpha}$. This is a limit result, but as we saw from the example above, the approximation with $n=6$ seems already quite precise. \Cref{F:logitps} illustrates that the convergence is indeed fast. It shows that while for small numbers of players, there is a difference between the lower bound (simultaneous decisions) and the upper bound (sequential decisions), the difference shrinks quickly and becomes negligible with 5--10 players. In particular, the figure illustrates that $\frac{1}{n} \left[ P^* - C - \frac{n}{\alpha} \right] \approx 0$ for any network with about ten players or more.
\begin{figure}
  \centering
  \includegraphics[trim={0.13cm 0.3cm 0.75cm 0.2cm},clip,width=0.5\linewidth]{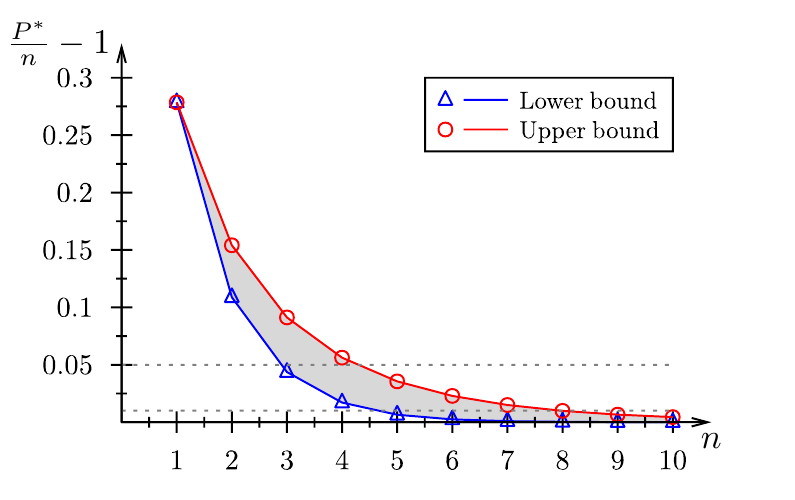}
  \caption{
    Bounds for equilibrium prices with logit demand $D(P) = \frac{e^{-P}}{1+e^{-P}}$ and $C=0$ depending on the number of firms.
  }
  \label{F:logitps}
\end{figure}

\section{Discussion} \label{S:discussion}

This paper characterizes the equilibrium behavior for a general class of price-setting games on a network. Under regularity assumptions, there is a unique equilibrium, which is straightforward to compute even with arbitrary demand function and complex networks. For the most common demand functions, such as linear, power, and logit demand, I provide even simpler characterization results.

The key distortion is the multiple-marginalization, which leads to too high markups both for efficiency and joint profit maximization. The marginalization problem increases with the number of firms but is magnified by strategic interactions. Firms set too high markups, not only because they do not internalize the negative impact on consumer surplus and other firms' profits, but also because they benefit from discouraging the other firms to set high markups.

The results define a natural measure of influentiality that ranks firms according to their markups and profits. Firms are more influential if they influence more firms or more influential firms. In some special cases, the influentiality measure simplifies to standard measures of centrality. I give examples, where it takes the form of Bonacich centrality, degree centrality, or is independent of the network structure.

Although the results are quite general in terms of network structure and demand functions, there are significant simplifying assumptions in other dimensions. First, I assume constant marginal costs, which simplifies the characterization but is not crucial for the analysis. Second, I model the competition in an extreme way---firms are either monopolists or price-takers. In some sense, this covers a few intermediate cases, where firms may be monopolists in some range, but when prices become too high, they become price-takers. But it would be certainly interesting to study other forms of imperfect competition, repeated interactions, bargaining, and more complex contract structures than posted prices. 

In this paper, the analysis is described in terms of price setting on a supply-chain network that supplies a single final product. There are other applications fitting the same mathematical model. An obvious example is multiple monopolists sell perfect complements. More generally, the model applies whenever multiple players choose actions, so that their payoffs depend linearly on their own actions, the marginal benefit is a decreasing function of the total action, and the actions are (higher-order) strategic substitutes. For example, the private provision of public goods and contests satisfy this general description.

Finally, the results have significant policy implications, which I have not discussed so far. In the following, I describe two simple examples to highlight two important policy implications. First, when analyzing mergers and acquisitions, it is crucial for the regulator to consider how the network of influences is affected. Second, in trade policy, small increases in any tariffs typically hurt all players, but more influential firms are harmed the most. Moreover, when considering non-marginal changes in trade policy, it is again crucial to consider the impact on the network as a whole.

\subsection{Example: Mergers} \label{SS:mergers}

To highlight how merger policy can be affected by changes in the network of influences, consider the following example with five monopolists. Firm $1$ produces raw material that is an input for two intermediate good producers, firms $2$ and $3$. Then final good producer firm $4$ uses inputs from firms $1$, $2$, and $3$ to produce the final good and sells it to final consumer through a retailer, firm $5$. The material flows are illustrated in \cref{F:questionablemerger2sc_pre}. However, it is important to specify how firms influence each other. Suppose that firms $2$ and $4$ are influenced by firm $1$. Firm $3$ makes the choice independently of firm $1$ and influences firm $4$. Finally, firm $5$ makes its choice independently of other firms. This is illustrated by \cref{F:questionablemerger2_pre}. 
To make the example more concrete, suppose that there are no costs and the demand function is $D(P) = (1-P)^4$, which allows us to use explicit formulas from \cref{SS:power}.

Suppose now that firms $1$ and $2$ would like to merge, leading to a new supply chain illustrated by \cref{F:questionablemerger2sc_m2}. Should the competition authority approve the merger? There are a few important aspects that the policymaker may consider. First, how does it affect the competition? By assumptions, we are analyzing the production of a single product with fixed demand function and monopolistic input providers, so the competition remains unaffected. Second, does it lead to cost-reductions or synergies? Again, we assume that there are no costs, so this remains unaffected. Third, there is one less firm, which reduces the marginalization problem. Combining these arguments, conventional wisdom suggests that the merger is socially desirable.
\begin{figure}[!ht]
  \centering
  \subfloat[Pre-merger material flows]{\scalebox{0.85}{
    \tikzset{>=stealth',every on chain/.append style={join},
  every join/.style={->}}  
\begin{tikzpicture}
[
cnode/.style={draw=black,fill=#1,minimum width=3mm,circle},
]
\node[cnode=white] (n1) at (0,0)     {1};
\node[cnode=white] (n2) at (1.3,1)   {2};            
\node[cnode=white] (n3) at (1.3,-1)  {3};
\node[cnode=white] (n4) at (2.6,0)   {4};
\node[cnode=white] (n5) at (4.2,0)   {5};

\draw[->,dashed] (n1) -- (n2);
\draw[->,dashed] (n1) -- (n3);
\draw[->,dashed] (n1) -- (n4);
\draw[->,dashed] (n2) -- (n4);
\draw[->,dashed] (n3) -- (n4);
\draw[->,dashed] (n4) -- (n5);
\end{tikzpicture}%

  } \label{F:questionablemerger2sc_pre}}
  \qquad
  \subfloat[Pre-merger influences]{\scalebox{0.85}{
    \tikzset{>=stealth',every on chain/.append style={join},
  every join/.style={->}}  
\begin{tikzpicture}
[
cnode/.style={draw=black,fill=#1,minimum width=3mm,circle},
]
\node[cnode=white] (n1) at (0,0)     {1};
\node[cnode=white] (n2) at (1.3,1)   {2};            
\node[cnode=white] (n3) at (1.3,-1)  {3};
\node[cnode=white] (n4) at (2.6,0)   {4};
\node[cnode=white] (n5) at (4.2,0)   {5};

\draw[->,black] (n1) -- (n2);
\draw[->,black] (n1) -- (n4);
\draw[->,black] (n3) -- (n4);
\end{tikzpicture}%
  } \label{F:questionablemerger2_pre}}
  \\
  \subfloat[Scenario A material flows]{\scalebox{0.85}{
    \tikzset{>=stealth',every on chain/.append style={join},
  every join/.style={->}}  
\begin{tikzpicture}
[
cnode/.style={draw=black,fill=#1,minimum width=3mm,circle},
]
\node[cnode=white] (n1) at (0.3,0.5) {1+2};
\node[cnode=white] (n3) at (1.3,-1)  {3};
\node[cnode=white] (n4) at (2.6,0)   {4};
\node[cnode=white] (n5) at (4.2,0)   {5};

\draw[->,dashed] (n1) -- (n4);
\draw[->,dashed] (n1) -- (n3);
\draw[->,dashed] (n3) -- (n4);
\draw[->,dashed] (n4) -- (n5);
\end{tikzpicture}%
  } \label{F:questionablemerger2sc_m1}}
  \qquad
  \subfloat[Scenario A influences]{\scalebox{0.85}{
    \tikzset{>=stealth',every on chain/.append style={join},
  every join/.style={->}}  
\begin{tikzpicture}
[
cnode/.style={draw=black,fill=#1,minimum width=3mm,circle},
]
\node[cnode=white] (n1) at (0.3,0.5) {1+2};
\node[cnode=white] (n3) at (1.3,-1)  {3};
\node[cnode=white] (n4) at (2.6,0)   {4};
\node[cnode=white] (n5) at (4.2,0)   {5};

\draw[->,black] (n1) -- (n4);
\draw[->,black] (n3) -- (n4);
\end{tikzpicture}%
  } \label{F:questionablemerger2_m1}}
  \\
  \subfloat[Scenario B material flows]{\scalebox{0.85}{
    \tikzset{>=stealth',every on chain/.append style={join},
  every join/.style={->}}  
\begin{tikzpicture}
[
cnode/.style={draw=black,fill=#1,minimum width=3mm,circle},
]
\node[cnode=white] (n1) at (0.3,0.5) {1+2};
\node[cnode=white] (n3) at (1.3,-1)  {3};
\node[cnode=white] (n4) at (2.6,0)   {4};
\node[cnode=white] (n5) at (4.2,0)   {5};

\draw[->,dashed] (n1) -- (n4);
\draw[->,dashed] (n1) -- (n3);
\draw[->,dashed] (n3) -- (n4);
\draw[->,dashed] (n4) -- (n5);
\end{tikzpicture}%
  } \label{F:questionablemerger2sc_m2}}
  \qquad
  \subfloat[Scenario B influences]{\scalebox{0.85}{
    \tikzset{>=stealth',every on chain/.append style={join},
  every join/.style={->}}  
\begin{tikzpicture}
[
cnode/.style={draw=black,fill=#1,minimum width=3mm,circle},
]
\node[cnode=white] (n1) at (0.3,0.5) {1+2};
\node[cnode=white] (n3) at (1.3,-1)  {3};
\node[cnode=white] (n4) at (2.6,0)   {4};
\node[cnode=white] (n5) at (4.2,0)   {5};

\draw[->,black] (n1) -- (n4);
\draw[->,black] (n3) -- (n4);
\draw[->,black] (n1) -- (n3);
\end{tikzpicture}%
  } \label{F:questionablemerger2_m2}}
  \caption{Examples: merger scenarios for firms $1$ and $2$. \label{F:questionablemerger2}}%
\end{figure}
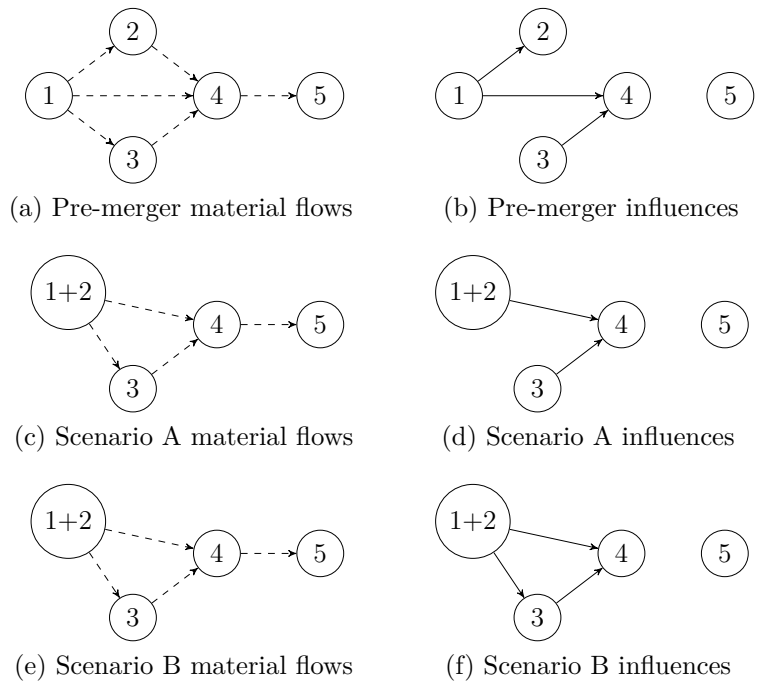

Using results from \cref{SS:power} to computing the pre-merger price of the final good gives us $P^* \approx 0.5897$. The corresponding profits, consumer surplus ($CS$) and social welfare ($SW = CS + \sum_i \pi_i$) are given in \cref{T:mergers}. Let us compare it with the simplest post-merger scenario (let us call it scenario A, see \cref{F:questionablemerger2_m1}), where the newly merged firm $1+2$ continues to influence only firm 4. Then the new price of the final good is $P^A \approx 0.5294$. This is clearly better for both the social welfare and for the total profit, due to a reduction in marginalization. Indeed, \cref{T:mergers} shows that the post-merger 1 scenario gives higher consumer surplus, total profit, social welfare, and all non-merging firms. But the joint profit of the merging firms $1$ and $2$ decreases from $0.0073$ to $0.0072$. The reason for this is that although the total profit has increased, the sum of the influences of the two firms decreased with the merger. This decrease in influentiality is large enough to make this merger undesirable for them. In a situation where they asked for permission to merge, this is probably not the right scenario to consider.
\begin{table}[ht!]
 \centering
   \begin{tabular}{lc|ccccc|cc}
    Scenario & $P^*$ & $\sum_i \pi_i^*$ & $\pi_1^*+\pi_2^*$ & $\pi_3^*$ & $\pi_4^*$ & $\pi_5^*$ & $CS$ & $SW$  \\
    \hline
    \multicolumn{9}{c}{(1) Demand function $D(P) = (1-P)^4$} \\
    Pre-merger & 
    $0.5897$ &  
      $0.0167$ & $0.0073$ & $0.0036$ & $0.0029$ & $0.0029$ &
      $0.0023$ & $0.0190$ \\
    Scenario A & 
    $0.5294$ & 
      $0.0260$ & $0.0072$ & $0.0072$ & $0.0058$ & $0.0058$ &
      $0.0046$ & $0.0306$ \\
    Scenario B & 
    $0.5461$ & 
      $0.0232$ & $0.0075$ & $0.0060$ & $0.0048$ & $0.0048$ &
      $0.0039$ & $0.0270$ \\
    \hline
    \multicolumn{9}{c}{(2) Demand function $D(P) = (1-P)^{\frac{4}{3}}$} \\
    Pre-merger & 
    $0.8447$ &  
      $0.0705$ & $0.0341$ & $0.0170$ & $0.0097$ & $0.0097$ &
      $0.0056$ & $0.0761$ \\
    Scenario B & 
    $0.8363$ & 
      $0.0749$ & $0.0337$ & $0.0192$ & $0.0110$ & $0.0110$ &
      $0.0063$ & $0.0812$ \\
    \hline
    \multicolumn{9}{c}{(3) Demand function $D(P) = (1-P)^{\frac{3}{4}}$} \\
    Pre-merger & 
    $0.9231$ &  
      $0.1348$ & $0.0699$ & $0.0350$ & $0.0150$ & $0.0150$ &
      $0.0064$ & $0.1412$ \\
    Scenario B & 
    $0.9288$ & 
      $0.1281$ & $0.0713$ & $0.0306$ & $0.0131$ & $0.0131$ &
      $0.0056$ & $0.1337$ \\
    \hline
    \multicolumn{9}{c}{(4) Demand function $D(P) = (1-P)^5$} \\
    Pre-merger & 
    $0.5283$ &  
      $0.0123$ & $0.0053$ & $0.0026$ & $0.0022$ & $0.0022$ &
      $0.0018$ & $0.0142$ \\
    Scenario A & 
    $0.4681$ & 
      $0.0199$ & $0.0054$ & $0.0054$ & $0.0045$ & $0.0045$ &
      $0.0038$ & $0.0237$ \\
    \hline
\end{tabular}
\caption{Examples of merger scenarios: 
  Panel 1: Scenario B is a Pareto improvement compared to the pre-merger situation. Scenario A is socially desirable, but firms 1 and 2 would choose not to merge. Panel 2: Scenario B would still be socially desirable, but would not happen voluntarily. Panel 3: Scenario B is socially undesirable, but would happen. Panel 4: Even scenario B is a Pareto improvement.
} \label{T:mergers}
\end{table}

Suppose alternatively that after the merger, the new firm $1+2$ becomes more influential, so that it can also influence the decision of firm $3$, as illustrated by \cref{F:questionablemerger2_m2}. In this case, the equilibrium price of the final good is $P^B \approx 0.5461$, which is still lower than the pre-merger price, so that this merger is socially desirable. Moreover, the joint profit of the merging firms is now $0.0075$, which is larger than the sum of their pre-merger profits. As \cref{T:mergers} illustrates, this scenario is, in fact, a strict Pareto-improvement compared to pre-merger---consumer surplus, profits of all firms (including non-merging), and social welfare strictly increase. 
The reason is that the additional influences (one direct and one indirect) allow the new firm $1+2$ capture a larger share of the surplus and this makes the merger desirable for them.

Note that these conclusions depend on the details of the network structure as well as the demand function. In scenario B, there is one less firm, the same number of direct influences, and one more indirect influence than in the pre-merger case. Whether or not it is socially preferable depends on the weights on different levels of influences. The particular demand function $D(P)=(1-P)^4$ is a power demand where the weights are declining at rate $\beta = \frac{1}{4}$. That is, indirect influences affect marginalization less than direct influences and therefore the merger is still socially desirable. Similarly, the joint influentiality of firms 1 and 2 depends on the weights.
Panel 2 in \cref{T:mergers} illustrates that when the demand function is $D(P)=(1-P)^{\frac{4}{3}}$, i.e.\ $\beta=\frac{3}{4}$, which puts more weight on the indirect weights, then the conclusions change. The merger is still socially desirable, but now the new firm 1+2 does not have enough influence to make the merger desirable.
Panel 3 in \cref{T:mergers} shows that if $D(P)=(1-P)^{\frac{3}{4}}$, so that $\beta=\frac{4}{3}>1$ and therefore indirect influences have larger weight than direct influences, this merger is not socially desirable anymore. However, in this case, the influentiality of the new firms is so large that they would prefer to merge. Finally, panel 4 shows that if $D(P)=(1-P)^5$, then the weight on indirect influences is so low that even a merger without any additional influences (scenario A) is a Pareto improvement compared to the pre-merger situation.

This example shows that even in simple cases the changes in the network of influences may have drastic policy implications. Note that the example is simple enough so that the equilibria can be computed directly. \Cref{T:characterization} provides a characterization for all possible scenarios when the demand function and the network may be more complex.

\subsection{Example: Tariffs} \label{SS:trade}

The previous discussion already illustrates the importance of taking the changes in the network of influences into account when considering policy decisions. The same message applies to trade policies as well. Changes in tariffs or quotas, as well as any trade restrictions, influence the supply chains and the interactions of firms within a supply chain. Therefore they naturally impact the network of influences. The results in this paper provide a tool to compare the outcomes under different scenarios.

In the simplest case when changes in tariffs are marginal so that the network of influences is unchanged, \cref{T:characterization} provides specific predictions. In particular, let us assume that the marginal cost of each input is $c_i = \hat{c}_i + t_i$, where $\hat{c}_i$ is the physical cost and $t_i$ is the tariff on good $i$. Then changes in tariffs can be thought of as changes in $\boldsymbol{c}=(c_0,c_1,\dots,c_n)$. 

When the changes are marginal, only the total marginal cost $C = \sum_{i=0}^n c_i$ affects the equilibrium price, profits, consumer surplus, and total welfare. Moreover, the equilibrium price is increasing and all payoffs are decreasing with $C$. This is easy to see from \cref{E:equilibrium}. Differentiating the equation gives 
\begin{equation}
  P^* - \sum_{k=1}^n \vone' \madj^{k-1} \vone g_k(P^*) = C
  \;\;\Rightarrow\;\;
  \dd{P^*}{C} 
  = \frac{1}{1 - \sum_{k=1}^n \vone' \madj^{k-1} \vone g_k'(P^*)} > 0.
\end{equation}
as each $g_k'(P^*) \leq 0$. Therefore only the aggregate changes in tariffs affect the equilibrium price of the final good. Clearly, consumer surplus depends only on the price of the final good. 
Although individual prices $p_i^*$ are affected by individual tariffs, the individual equilibrium profits $\pi_i(\bp^*) = \infl_i(\madj) D(P^*)$ are only affected by tariffs through its impact on the price of the final good. Any increase in the sum of tariffs leads to an increase in the price of the final good and therefore decrease in profits that are proportional to influentiality measure $\infl_i(\madj)$. So, the more influential firms are affected more by the tariffs, regardless of which individual tariffs or subsidies are imposed.
Finally, defined as the sum of consumer surplus, all profits, and the tariff revenues are marginally affected as
\begin{equation}
  TW
  = \underbrace{\int_{P^*}^{\infty} D(P) dP}_{=\text{Consumer Surplus}} 
  + \underbrace{D(P^*) (P^*-\hat{C}-T)}_{=\sum_i \pi_i^*}
  + D(P^*) T.
\end{equation}
Direct effects of tariffs cancel out as profits are reduced exactly by the tariff revenue. The only effect is through the change in the price of the final good, which is increasing in tariffs. Differentiating total welfare with respect to the equilibrium price of the final good gives
\begin{equation}
  \dd{TW}{P^*}
  = 
  -D(P^*) 
  + D'(P^*) (P^*-\hat{C})
  + D(P^*)
  = 
  D'(P^*) (P^*-\hat{C})
  < 0
  \iff 
  P^* > \hat{C}.
\end{equation}
This is a standard textbook finding implying that the socially optimal tax on a monopoly is, in fact, a subsidy that equalizes price with marginal cost. 

Note that there are important aspects missing from this simple application of the model. The main goal of using tariffs and other trade policies is to affect trade flows. This changes the supply chain network and the network of influences. As highlighted in the discussion above, it may have a large impact both on the consumer surplus and the profits and should be therefore considered in any such policy evaluation.

\bibliography{zotero_pricing_networks}
\bibliographystyle{ecta.bst}

\appendix

\section{Proofs} \label{A:proofs}

\subsection{Proof of \texorpdfstring{\cref{E:completelymonotonefunctions}}{lemma 1}}

\begin{proof}
  In each case, I directly verify the claim:
  \begin{enumerate}
    \item Linear demand is a special case of power demand with $d=\beta=1$. 
    \item Power demand implies $g(P) 
    = - \frac{d (a-bP)^{\frac{1}{\beta}}}{d \frac{1}{\beta} (a-bP)^{\frac{1}{\beta}-1} (-b)} 
    = \beta (\oP-P)$, where $\oP=\frac{a}{b}$. Then $-g'(P) = \beta > 0$ and $(-1)^k \dd[k]{g(P)}{P} = 0$, for all $k > 1$.
    \item Logit demand implies $g(P) 
    = \frac{1}{\alpha} \left[ 1+e^{-\alpha P} \right]
    $. Then $(-1)^k \dd[k]{g(P)}{P} = \alpha^{k-1} e^{-\alpha P} > 0$.
    \item Exponential demand implies $g(P)
    = \frac{1}{\alpha} \left[
    \oP e^{-\alpha P} - 1
    \right]
    $. Therefore $(-1)^k \dd[k]{g(P)}{P} = \alpha^{k-1} \oP e^{-\alpha P} > 0$.
  \end{enumerate}
\end{proof}

\subsection{Proof of \texorpdfstring{\cref{T:characterization}}{theorem 1}}

Before the proof, let me introduce some useful notation.
Each player $i \in \gN = \{1,\dots,n\}$, observes prices of some players. Let the set of these players be $\Nobs_i = \{j : a_{ji}=1 \} \subset \gN$ (possibly empty set) and vector of these prices $\pobs_i = (p_j)_{j \in \Nobs_i}$. Player $i$'s strategy is $p_i^*(\pobs_i)$. 
Player $i$ also influences some players, let the set of these players be $\Ninf_i = \{ j : a_{ij}=1 \} \subset \gN$ (again, possibly empty). Each such player $j \in \Ninf_i$ uses the equilibrium strategy $p_j^*(\pobs_j)$. By definition, $i \in \Nobs_j$, i.e.\ $p_i$ is one of the inputs in $\pobs_j$. However, $i$ does not necessarily observe all prices in $\pobs_j$, therefore it must make an equilibrium conjecture about these values. Let $p_j^i(p_i,\pobs_i)$ denote player $j$'s action as seen by player $i$. That is, $p_j^i(p_i,\pobs_i)=p_j^*(\pobs_j')$, where $\pobs_j = (p_k')_{k\in\Nobs_j}$ is such that $p_k'=p_k$ if $k \in \Nobs_i$ or $k=i$ and $p_k'=p_k^i(p_i,\pobs_i)$ otherwise. The last step makes the definition recursive, but it is well-defined, as each such step strictly reduces the number of arguments in the function.
Finally, there are also some players whose prices that $i$ neither observes nor influences, let this set be $\Nuno_j = \{ j : a_{ji}=a_{ij} = 0\} \subset \gN$. For these players, $i$ expects the actions to be $p_j^i(\pobs_i)$ defined in the same way as above, but its arguments do not include $p_i$.

Using this notation, a firm $i$ that observes $\pobs_i$ and sets its price to $p_i$, expects the price of the final good to be
\begin{equation} \label{E:Pi}
  P^i(p_i|\pobs_i) 
  = c_0 
  + p_i
  + \sum_{j \in \Nobs_i} p_j
  + \sum_{j \in \Ninf_i} p_j^i(p_i,\pobs_i)
  + \sum_{j \in \Nuno_i} p_j^i(\pobs_i).
\end{equation}

The main idea in the proof is the following. Instead of choosing price $p_i$ to maximize profit $(p_i-c_i) D(P^i(p_i|\pobs_i))$, we can think of player $i$ choosing the final good price $P$ to induce. For this, let me assume that in the relevant range, $P^i(p_i|\pobs_i)$ is smooth and strictly increasing in $p_i$, so that it has a differentiable and strictly increasing inverse function $f_i(P|\pobs_i)$ such that $P^i(f_i(P|\pobs_i)|\pobs_i) = P$. Then the maximization problem is
\[
  \max_P [f_i(P|\pobs_i)-c_i] D(P),
\]
which leads to first-order condition $f_i'(P|\pobs_i) D(P)
  +
  [f_i(P|\pobs_i)-c_i] D'(P) = 0$ or equivalently
\begin{equation} \label{E:FOCi}
  f_i(P|\pobs_i)-c_i  
  = g(P) f_i'(P|\pobs_i).
\end{equation}
Note that there is one-to-one mapping between representing equilibrium behavior in terms of functions $f_i(P|\pobs_i)$ and in terms of $p_i^*(\pobs_i)$.

\begin{proof}
  Observe that the equilibrium must be interior, i.e.\ each $c_i < p_i < \oP$ for each firm. If this is not the case for the firm $i$, then its equilibrium profit is non-positive. This could be for one of two reasons. First, the equilibrium price of the final good is so high that $D(P)=0$. In this case, all equilibrium profits are non-positive and there must be at least one firm $i$ who, by reducing its price (and anticipating the responses of firms influenced), can make the final good price low enough so that it ensures a strictly positive profit. This would be a profitable deviation. Second, if $P<\oP$ and $p_i \leq c_i$, then firm $i$ can raise its price slightly and increase its profit. 
  
  I will first derive necessary conditions for an interior equilibrium and combine them into one necessary condition, which gives \cref{E:equilibrium}. I then show that it has a unique solution and finally verify that it is indeed an equilibrium by verifying that each firm indeed chooses a price that maximizes its profit.

  Let us start with any player $i$ who does influence any other players, i.e.\ $\ve_i' \madj \vone = 0$ or equivalently $\Ninf_i = \varnothing$. Then we can rewrite \cref{E:Pi} as
  \begin{equation}
    P = c_0 + f_i(P|\pobs_i) 
    + \sum_{j \in \Nobs_i} p_j
    + \sum_{j \in \Nuno_i} p_j^i(\pobs_i).
  \end{equation}
  Differentiating this expression with respect to $P$ shows that $f_i'(P|\pobs_i) = 1$ (that is, player $i$ can raise the price of the final good by $\varepsilon$ by raising its own price by $\varepsilon$). Therefore \cref{E:FOCi} implies $f_i(P|\pobs_i) = c_i+g(P)$. 
  Note that this expression is independent of $\pobs_i$, so I can drop it as an argument for $f_i$ and write simply as $f_i(P)=c_i+g(P)$.
  
  Let us take now any player $i$ and suppose that the optimal behavior of all players $j \in \Ninf_i$ is described corresponding functions $f_j(P)$ that do not depend on the remaining arguments $\pobs_j$. Then we can rewrite \cref{E:Pi} as
  \begin{equation}
    P 
      = c_0 
      + f_i(P|\pobs_i)
      + \sum_{j \in \Nobs_i} p_j
      + \sum_{j \in \Ninf_i} f_j(P)
      + \sum_{j \in \Nuno_i} p_j^i(\pobs_i).
  \end{equation}
  Differentiating this expression and inserting it to \cref{E:FOCi} gives
  \begin{equation}
     f_i'(P|\pobs_i)
     = 1 - \sum_{j \in \Ninf_i} f_j'(P)
     \;\;
     \Rightarrow
     \;\;
     f_i(P|\pobs_i)
     = g(P) \left[
       1 - \sum_{j \in \Ninf_i} f_j'(P)
     \right].
  \end{equation}
  This expression is again independent of the arguments $\pobs_i$, which we can therefore drop.
  Moreover, these arguments give precise analytic expressions for $f_i(P)$ functions. We already saw that $f_i(P)=g(P)=\sum_{k=1}^n \ve_i' \madj^{k-1} \vone g_k(P)$ when $\ve_i' \madj^{k-1} \vone = 0$ for all $k>1$ (i.e.\ players who do not influence anybody). Suppose that every player $j \in \Ninf_i$ has
  \begin{equation} \label{E:fjdef}
    f_j(P) - c_j= \sum_{k=1}^n \ve_j' \madj^{k-1} \vone g_k(P).
  \end{equation}
  Then for player $i$ we must have
  \begin{equation}
    f_i(P) - c_i
     = g(P) \left[
       1 - \sum_{j \in \Ninf_i} f_j'(P)
     \right]
     = \underbrace{g(P)}_{= \ve_i' \madj^0 \vone g_1(P)}
     + \sum_{k=1}^n \underbrace{\sum_{j \in \Ninf_i}  \ve_j' \madj^{k-1} \vone}_{=\ve_i' \madj^k \vone} \underbrace{[-g_k'(P)g(P)]}_{g_{k+1}(P)},
  \end{equation}
  which, after change of variables from $k$ to $k-1$ and combining the terms, gives the same expression as in \cref{E:fjdef}.\footnote{Note that no player can have level-$n$ influences, i.e.\ $\ve_i' \madj^n \vone = 0$.}  
  
  Therefore on-path, when the equilibrium price of the final good is $P^*$, the individual prices are indeed given by the expressions in the theorem. The price of the final good must be sum all input prices, therefore $P^*$ must satisfy
  \[
    P^*
    = c_0 + \sum_{i \in \gN} f_i(P^*)
    = \underbrace{c_0 + \sum_{i \in \gN} c_j}_{=C} 
    + \sum_{k=1}^n \underbrace{\sum_{i \in \gN} \ve_i' \madj^{k-1} \vone}_{=\vone' \madj^{k-1} \vone} g_k(P^*),
  \]
  which gives the \cref{E:equilibrium}.
  
  Below I prove two technical lemmas (\cref{L:gkmonotone,L:ffimonotone}) provide monotonicity properties that imply existence and uniqueness of equilibria. 
  We can rewrite \cref{E:equilibrium} as $f(P)=P-C-\sum_{k=1}^n \vone' \madj^{k-1} \vone g_k(P) = 0$. At $P=0$ we have $f(0)=-C-\sum_{k=1}^n \vone' \madj^{k-1} \vone g_k(0) < 0$ and $\lim_{P \to \oP} f(P) > 0$. By \cref{L:ffimonotone}, function $f(P)$ is strictly increasing and therefore $f(P)=0$ has a unique solution, which is the equilibrium price of the final good $P^* \in (0,\oP)$.
  
  Next, in the argument above, we assumed that the inverse function of $f_i(P)$ function is strictly increasing. The construction implied a necessary condition that $f_i(P)$ must satisfy and \cref{L:ffimonotone} shows that it implies that $f_i(P)$ is indeed strictly increasing, therefore the inverse function $P^i(p_i|\pobs_i)$ is indeed a well-defined strictly increasing function.
  Finally, to verify that the solution we found is indeed an equilibrium, we need to verify that the solution we derived is indeed a global maximizer for each firm. Notice that by \cref{L:ffimonotone}, the optimality condition \cref{E:FOCi} has a unique solution for each firm. Therefore we have identified a unique local optimum for each firm. As we already verified that corner solutions would give non-positive profits for each firm and the interior solution gives strictly positive profit, this must be a global maximizer. 

  \begin{lemma}[Monotonicity of $g_k(P)$] \label{L:gkmonotone}
    $g_k(P)$ is $(d(\madj)+1-k)$-times monotone.
  \end{lemma}
  \begin{proof}
    $g_1(P) = g(P) = -\frac{D(P)}{D'(P)}$ is $d(\madj)$-times monotone by \cref{A:ass_monotone}. 
    Therefore $g'(P)$ is $(d(\madj)-1)$-times and $g_2(P) = -g_1'(P) g(P)$ is $(d(\madj)-1)$-times monotone. The rest follows by induction in the same way, if $g_k(P)$ is $(d(\madj)+1-k)$-times monotone, then $g_{k+1}(P) = -g_k'(P) g(P)$ is $(d(\madj)-k)$-times monotone.
  \end{proof}
  \begin{lemma}[Monotonicity of $f(P), f_i(P)$] \label{L:ffimonotone}
    The following monotonicity properties hold
    \begin{enumerate}
        \item $f(P) = P - C - \sum_{k=1}^n \vone' \madj^{k-1} \vone g_k(P)$ is strictly increasing,
        \item $f_i(P) = c_i - \sum_{k=1}^n \ve_i' \madj^{k-1} \vone g_k(P)$ is strictly increasing for each $i \in \{1,\dots,n\}$,
        \item $f_i'(P) g(P) = \sum_{k=1}^n \ve_i' \madj^{k-1} \vone g_k(P)$ is (weakly) decreasing for each $i \in \{1,\dots,n\}$.
    \end{enumerate}    
  \end{lemma}
  \begin{proof}
    Each $\vone' \madj^{k-1} \vone$ and $\ve_i' \madj^{k-1} \vone$ is a non-negative integer and each $g_k(P)$ weakly decreasing in $-P$ by \cref{L:gkmonotone}, which implies weak monotonicity of $f_i'(P) g(P)$.
    Moreover, when $k=1$, then $g_1(P)=g(P)$ which is strictly decreasing by \cref{A:ass_monotone} and $\ve_i' \madj^{k-1} \vone = 1>0$, which implies that $f_i(P)$ is strictly increasing. As $P-C$ is strictly increasing, then $f(P)$ is also strictly increasing.
  \end{proof}
\end{proof}

\subsection{Proof of \texorpdfstring{\cref{L:logit}}{lemma 2}}

\begin{proof}
  Using the facts that $g(P)=\frac{1}{\alpha} \left[1+e^{-\alpha P} \right] > \frac{1}{\alpha}$ and $g_k(P)>0$ for all $k>0$, \cref{E:equilibrium} gives
    $P^*
      = C+\sum_{k=1}^n \vone' \madj^{k-1} \vone g_k(P^*) 
      > C+n g(P^*)
      > C+\frac{n}{\alpha}$.
  Similarly for individual prices, $p_i^*
    = c_i + \sum_{k=1}^n \ve_i' \madj^{k-1} \vone g_k(P^*)
    > c_i + g(P^*)
    > c_i + \frac{1}{\alpha}$.
    
  Using the lower bound for $P^*$, we can bound $e^{-\alpha P^*} < 
  e^{-\alpha C - \alpha \frac{n}{\alpha}}
  =
  e^{-\alpha C} e^{- n}$. Therefore $e^{-\alpha P^*} = O(e^{-n})$.
  I use this result to prove \cref{E:logitgks} that shows that $g_1(P^*) = \frac{1}{\alpha} + O(e^{-n})$ and $g_k(P^*) = O(e^{-n})$ for all $k>1$. Therefore \cref{E:equilibrium} gives
  \begin{equation}
    P^*
    = C+\sum_{k=1}^n \vone' \madj^{k-1} \vone g_k(P^*)
    = C+\frac{n}{\alpha} + O(e^{-n}) \bonacich(\madj),
  \end{equation}
  where $\bonacich(\madj) = \sum_{k=1}^n \vone' \madj^{k-1} \vone$.
  Now, note that $\bonacich(\madj)$ increases each time an edge is added to $\madj$, so its upper bound is when the network is most connected (fully sequential decisions) and lower bound with least connected network (simultaneous decisions), so that $n \leq \bonacich(\madj) \leq 2^n-1$. Therefore $\bonacich(\madj) = O(2^n)$. Inserting this observation to previous expression gives
  $P^* = C + \frac{n}{\alpha} + O\left( \left[ \frac{2}{e} \right]^{n} \right)$.
  Finally, for the equilibrium expression for individual prices is
  \begin{equation}
    p_i^*  = c_i 
    + \sum_{k=1}^n \ve_i' \madj^{k-1} \vone g_k(P^*)
    = c_i + \frac{1}{\alpha} + O(e^{-n}) \bonacich_i(\madj),
  \end{equation}
  where $B_i(\madj) = \sum_{k=1}^n \ve_i' \madj^{k-1} \vone$, which is by the same arguments as above $\bonacich_k(\madj) = O(2^n)$ and therefore $p_i = c_i + \frac{1}{\alpha} + O \left( \left[ \frac{2}{e} \right]^n \right)$.
\end{proof}

\begin{lemma} \label{E:logitgks}
With logit demand $D(P)=d \frac{e^{-\alpha}}{1+e^{-\alpha P}}$, functions $g_k(P)$ and their derivatives have the following limit properties at $P=P^*$
\begin{enumerate}
    \item $g_1(P^*) = \frac{1}{\alpha} + O(e^{-n}) = O(1)$, $g_k(P^*)=O(e^{-n})$ for all $k \in \{2,\dots,n\}$,
    \item $\dd[l]{g_k(P^*)}{P} = O(e^{-n})$ for all $k,l \in \{1,\dots,n\}$.
\end{enumerate}
\end{lemma}
\begin{proof}
Consider $g_1(P)$ first. We get $g_1(P^*) = g(P^*) = \frac{1}{\alpha} +\frac{1}{\alpha} e^{-\alpha P^*} = \frac{1}{\alpha} + O(e^{-n}) = O(1)$. The derivatives $\dd[l]{g_1(P^*)}{P} = (-\alpha)^{k-1} e^{-\alpha P^*} = O(e^{-n})$. The rest of the proof is by induction. Suppose that the claim holds for $g_1,\dots,g_k$. Now,
 $g_{k+1}(P^*) = - g_k'(P^*) g(P^*) = O(e^{-n})$ as $g(P^*) = O(1)$ and $g_k'(P^*)=O(e^{-n})$ by induction assumption. Each derivative
 \begin{equation}
   \dd[l]{g_{k+1}(P^*)}{P}
   = 
   - \sum_{j=0}^l \binom{l}{j} g_k^{(l - j+1)}(P^*)g^{(j)}(P^*)
 \end{equation}
 Each $g_k^{(l - j+1)}(P^*)=O(e^{-n})$ by induction assumption (as $l-j+1 \geq 1$). When $j=0$, then $g^{(j)}(P^*) = g(P^*) = O(1)$. Therefore the first element of the sum is $g_k^{(l - j+1)}(P^*)g^{(j)}(P^*) = O(e^{-n})$. 
 For all other elements $j > 0$, so the term $g^{(j)}(P^*)=O(e^{-n})$ and therefore each $g_k^{(l - j+1)}(P^*)g^{(j)}(P^*) = O(e^{-2n})$, which is dominated asymptotically by $O(e^{-n})$. This proves that $\dd[l]{g_{k+1}(P^*)}{P} = O(e^{-n})$.
\end{proof}

\end{document}